\documentclass[11pt,doublesided, doublecolumn]{IEEEtran}
\usepackage{algorithm}
\usepackage[noend]{algorithmic}
\usepackage[geometry]{ifsym}
\usepackage {ifsym}
\usepackage{amssymb}
\usepackage{yhmath}
\usepackage{amscd}
\usepackage{dsfont}
\usepackage{amsmath}
\usepackage{epsfig}
\usepackage{graphics}
\usepackage{psfrag}
\usepackage{rotating}
\usepackage{amsmath}
\usepackage{amsfonts}
\usepackage{url}
\usepackage{color}

\newtheorem{example}{Example}
\newtheorem{theorem}{Theorem}
\newtheorem{definition}[theorem]{Definition}
\newtheorem{proposition}[theorem]{Proposition}

\newcommand{\ind}[1]{\mathds{1}_{\left\lbrace #1 \right\rbrace}}

\newcommand{\bs}{\boldsymbol}

\newcommand{\ds}{\displaystyle}

\newcommand{\pr}[2]{\mathrm{Pr} \left(#1 \left| #2 \right. \right)}
\newcommand{\GameSF}{\wideparen{\mathcal{G}} = \left( \mathcal{K}, \left\lbrace \mathcal{A}_k \right\rbrace_{k \in \mathcal{K}},\left\lbrace f_{k}\right\rbrace_{ k \in\mathcal{K}} \right) }
\newcommand{\gameSF}{\wideparen{\mathcal{G}}}
\newcommand{\GameSFmixed}{\wideparen{\mathcal{G}}' = \left( \mathcal{K}, \left\lbrace \triangle\left(\mathcal{A}_k\right)\right\rbrace_{k \in \mathcal{K}}, \left\lbrace \bar{f}_k \right\rbrace_{k \in \mathcal{K}} \right)}
\newcommand{\GameSE}{\hat{\mathcal{G}} = \left( \mathcal{K}, \left\lbrace \mathcal{A}_k \right\rbrace_{k \in \mathcal{K}},\left\lbrace u_{k}\right\rbrace_{ k \in\mathcal{K}}, \left\lbrace f_{k}\right\rbrace_{ k \in \mathcal{K}}\right)}
\newcommand{\gameSE}{\hat{\mathcal{G}} }
\newcommand{\GameSEprime}{\hat{\mathcal{G}} = \left( \mathcal{K}, \left\lbrace \mathcal{A}_k \right\rbrace_{k \in \mathcal{K}},\left\lbrace u_{k}\right\rbrace_{ k \in\mathcal{K}}, \left\lbrace f_{k}' \right\rbrace_{ k \in \mathcal{K}}\right)}
\newcommand{\gameSEprime}{\hat{\mathcal{G}} }
\newcommand{\GameNF}{\mathcal{G} = \left( \mathcal{K}, \left\lbrace\mathcal{A}_k \right\rbrace_{k \in \mathcal{K}},\left\lbrace u_{k}\right\rbrace_{ k \in \mathcal{K}}\right)}
\newcommand{\gameNF}{\mathcal{G}}
\newcommand{\gameNFmodif}{\mathcal{G}}
\newcommand{\GameESF}{\tilde{\mathcal{G}} = \left(\mathcal{K}, \left\lbrace \mathcal{A}_k \right\rbrace_{k \in \mathcal{K}},\left\lbrace c_{k}\right\rbrace_{ k \in\mathcal{K}}, \left\lbrace f_{k}\right\rbrace_{ k \in \mathcal{K}}\right)}
\newcommand{\GameESFapp}{\mathcal{G} = \left(\mathcal{K}, \left\lbrace \mathcal{A}_k \right\rbrace_{k \in \mathcal{K}},\left\lbrace u_{k}\right\rbrace_{ k \in\mathcal{K}}, \left\lbrace g_{k}\right\rbrace_{ k \in \mathcal{K}}\right)}
\newcommand{\gameESFapp}{\mathcal{G}}

\newcommand{\gameESF}{\tilde{\mathcal{G}}}
\newcommand{\GameESFpotential}{\tilde{\mathcal{G}'} = \left(\mathcal{K}, \left\lbrace \mathcal{A}_k \right\rbrace_{k \in \mathcal{K}},\left\lbrace \phi \right\rbrace_{ k \in\mathcal{K}}, \left\lbrace f_{k}\right\rbrace_{ k \in \mathcal{K}}\right)}
\begin{document}
\title{Quality-Of-Service Provisioning in Decentralized Networks: A Satisfaction Equilibrium Approach}
\author{S.M.~Perlaza,    
        H.~Tembine,      
        S.~Lasaulce,     
        ~and~M.~Debbah   
\thanks{S.M. Perlaza is with the Alcatel-Lucent Chair in Flexible Radio at SUPELEC. $3$ rue Joliot-Curie, $91192$, Gif-sur-Yvette, cedex. France. (samir.medina-perlaza@supelec.fr)}
\thanks{H.~Tembine is with  the Telecommunications Department at SUPELEC. $3$ rue Joliot-Curie,  $91192$, Gif-sur-Yvette, cedex. France. (tembine@ieee.org)}
\thanks{S. Lasaulce is with the Laboratoire des Signaux et Systemes - LSS (CNRS-SUPELEC-Paris Sud). $3$ rue Joliot-Curie, $91192$, Gif-sur-Yvette, cedex. France. (lasaulce@lss.supelec.fr) }
\thanks{M. Debbah is with the Alcatel-Lucent Chair in Flexible Radio at SUPELEC. $3$ rue Joliot-Curie, $91192$, Gif-sur-Yvette, cedex. France. (merouane.debbah@supelec.fr)}%
}
\maketitle
\begin{abstract}
\boldmath
This paper introduces a particular game formulation and its corresponding notion of equilibrium, namely the satisfaction form (SF) and the satisfaction equilibrium (SE). A game in SF models the case where players are uniquely interested in the satisfaction of some individual performance constraints, instead of individual performance optimization. Under this formulation, the notion of equilibrium corresponds to the situation where all players can simultaneously satisfy their individual constraints. The notion of SE,
models the problem of QoS provisioning in decentralized self-configuring networks. Here, radio devices are satisfied if they are able to provide the requested QoS. Within this framework, the concept of SE is formalized for both pure and mixed strategies considering finite sets of players and actions. In both cases, sufficient conditions for the existence and uniqueness of the SE are presented.
When multiple SE exist, we introduce the idea of effort or cost of satisfaction and we propose a refinement of the SE, namely the efficient SE (ESE). At the ESE, all players adopt the action which requires the lowest effort for satisfaction. A learning method that allows radio devices to achieve a SE in pure strategies in finite time and requiring only one-bit feedback is also presented. 
Finally, a power control game in the interference channel is used to highlight the advantages of modeling QoS problems following the notion of SE rather than other equilibrium concepts, e.g., generalized Nash equilibrium.
\end{abstract}

\section{Introduction}\label{SecIntroduction}
Nowadays, in the context of decentralized networks, a well accepted idea among radio device manufacturers, network designers and service providers is to consider radio devices as smart entities able to autonomously adopt the transmit/receive configuration that allows the achievement of  certain level of individual quality-of-service (QoS).
In particular, a transmit/receive configuration can be described in terms of 
power allocation polices, coding-modulation schemes, scheduling polices, decoding order, \emph{etc}.
Within this framework, the resulting decentralized radio resource sharing problem can be analyzed using tools from both non-cooperative game theory (GT) \cite{Mackenzie-01, Lasaulce-Tembine-Book-2011} and variational inequality (VI) theory \cite{Scutari-Palomar-2009,Scutari-Palomar-2010}. Indeed, strong connections between both approaches have been already highlighted \cite{Facchinei-2007b}, \cite{Facchinei-Pang-2009}.  
From this point of view, the resulting competitive interaction through mutual interference between all radio devices can be modeled by a game in normal form with coupled action sets
(See \cite{Lasaulce-Tembine-Book-2011} and references therein). Here, radio devices are  players aiming to optimize an individual benefit (performance metric) and guarantee certain individual constraints (quality-of-service, QoS)  by selfishly choosing their actions, i.e., the transmit/receive configurations. In this context, the above mentioned coupling between the actions sets stems from the fact that the set of transmit/receive configurations that satisfy the QoS of a given radio device depend on the configurations adopted by all the other devices.
Under this game formulation, initially proposed in \cite{Debreu-1952}, \cite{Rosen-65}, the notion of Nash equilibrium (NE) \cite{Nash-1950} is transformed into the notion described in \cite{Debreu-1952}, which is known nowadays as generalized Nash equilibrium (GNE). The notion of GNE is already well integrated in the context of QoS provisioning in decentralized networks. For instance, it is used in \cite{PangScutari-IT-08} and \cite{Altman-Altman-2003} to model a network where radio devices aim to minimize the transmit power consumption while guaranteeing some individual minimum signal to interference plus noise ratios (SINR). Therein, the existence and uniqueness of the GNE is discussed. In particular, both articles provide sufficient conditions for the convergence of the best response dynamics \cite{Fudenberg-Tirole-1991} to the GNE, in the context of finite action sets  (only in \cite{Altman-Altman-2003}) and convex and closed action sets (in \cite{PangScutari-IT-08} and \cite{Altman-Altman-2003}). Other examples are provided in \cite{Yin-Mehta-2011} and references therein.

\noindent
Interestingly, in \cite{Facchinei-2007b}\cite{Facchinei-Pang-2009}, it has been shown that a GNE can be formulated as VI problem as long as the set of actions is closed and convex and the utility functions are continuously differentiable. From such a formulation, many conclusions regarding the existence and uniqueness of the equilibrium can be obtained based on existing results in VI theory \cite{Scutari-Palomar-2009}. Additionally, algorithms to achieve the GNE can be developed based on existing methods for solving the  VI problem, for instance, see \cite{Iiduka-2010}\cite{Iiduka-2011}. 
Nonetheless, the use of VI in the context of decentralized networks is limited since it requires the explicit expression of the metric projection onto the set of actions \cite{Yamada-2001,LuX-2009, MarinoG-2011}. This condition is particularly difficult to be satisfied in decentralized networks since radio devices possess only local information \cite{Rose-Perlaza-CommMag2011}. An extensive analysis of VI in the  context of decentralized wireless communications is presented in \cite{Scutari-Palomar-2009}.
 
\noindent
In this paper, we present a game formulation and its corresponding
notion of equilibrium, namely the satisfaction form (SF) and the satisfaction equilibrium (SE). The satisfaction form models the case where players are concerned by the satisfaction of their individual constraints but not the optimization of its performance metric.
A game is said to be in SE, when all players simultaneously satisfy their individual constraints.
In the sequel, we exclusively focus on the case of finite sets, since we aim to model the fact that radio devices choose from a finite number of transmit configurations. 

\noindent
The idea behind SE was originally introduced in \cite{Ross-06, Ross-07} for a particular class of conditions in pure strategies for the case of finite set of actions. Later, the concept was formulated in terms of a fixed point inclusion for the case of pure strategies in \cite{Perlaza-Globecom-2010}, in the context of both finite action sets and convex and closed action sets.  
The advantages of the notion of SE over the classical notions such as GNE, at least in the domain of signal processing for wireless communications are multifold. Here, we highlight the fact that, $(i)$ the existence of the SE in pure strategies is less restrictive than the notion of GNE. That is, a network can possess a SE but not a GNE. $(ii)$ The behavioral rules to learn a SE are simpler than behavioral rules to learn GNE \cite{Lasaulce-Tembine-Book-2011, Rose-Perlaza-CommMag2011}. 
Within this framework, the contributions presented in this paper are the followings:
\begin{itemize}
\item The notions of SF and
 SE are formalized in both pure strategies (PS) and mixed strategies (MS) in the context of finite games. Conditions for the existence of the SE in PS and MS are established. 
\item We introduce the notion of epsilon-satisfaction equilibrium ($\epsilon$-SE), which consists of a mixed strategy which allows all players to be satisfied with probability not less than a certain threshold. 
This equilibrium concept turns out to be less restrictive in terms of existence than the SE. 
 \item A refinement of the notion of SE to which we refer to as efficient SE (ESE) is presented as a mechanism for equilibrium selection involving the idea of \emph{effort for satisfaction}. 
\item A simple learning algorithm
to achieve SE, based on the algorithms proposed in  \cite{Ross-07}, is presented.  
\end{itemize}

\noindent
The sequel of the paper is organized as follows. In Sec. \ref{SecProblemFormulation}, the QoS  provisioning problem in decentralized self-configuring networks is formulated. In Sec. \ref{SecSE}, the notions of SF and SE are presented.
In Sec. \ref{SecExistenceAndUniqueness}, the existence and uniqueness of the SE are analyzed.  In Sec. \ref{SecSEandOtherEquilibria}, we compare the notions of SF and SE with existing equilibrium notions. 
In Sec. \ref{SecLearningSE}, behavioral rules that allow radio devices to learn a SE are described. In Sec. \ref{SecApplications}, the notion of SE is used in the context of a simple power control in the interference channel where transmitters must guarantee a minimum transmission rate.   
The paper is concluded by Sec. \ref{SecConclusions}.

\section{Problem Formulation}\label{SecProblemFormulation}

We consider a fully decentralized network where transmitters communicate with their respective receivers by sharing a common set of radio resources, and thus, subject to a competitive interaction. For instance, the usage of the same frequency bands brings a mutual interference condition. Within this framework, each radio device, either a transmitter or receiver, aims to autonomously adjust its transmit/receive configuration in order to guarantee a communication with certain QoS level. The QoS can be described by several parameters depending on the type of services. Some classical examples are constraints over the transmission rates and maximum delays. The key point is that the feasibility of the QoS of each radio device depends on the configurations adopted by all the other devices in the network.
Regarding the transmit/receive configurations, we assume that there exists a finite number of feasible choices for each radio device. In particular, a configuration can be described in terms of  channel selection and power allocation policy, modulation and error correction schemes, constellation sizes, etc. Similarly,  receivers might tune their scheduling, decoding order, etc.
 
\noindent
In this work, we consider that radio devices are selfish entities aiming to satisfy their QoS individual constraints. Each device is in particular careless of whether other devices achieve or not their required QoS. Moreover,
message exchanging between radio devices for establishing a sort of coordination to jointly improve the individual or global performance is not considered here. 
This is basically because of the amount of signaling it might require and also because devices are not necessarily able to communicate 
due to the use of different physical layer technologies.

\noindent
In the sequel of this article, we intend to provide a mathematical framework for tackling this problem using tools from game theory. Particular attention is given to real-system implementation constraints, for instance finite number of choices and limited numbers of bits dedicated to feedback. 

\section{Games in Satisfaction Form and Satisfaction Equilibrium}\label{SecSE}

In this section, we introduce a novel game formulation where in contrast to existing formulations (e.g., normal form \cite{Nash-1950} and normal form with constrained action sets \cite{Debreu-1952}), the idea of performance optimization, i.e., utility maximization or cost minimization, does not exist. In our formulation, to which we refer to as satisfaction-form, the aim of the players is to adopt any of the actions which allows them to satisfy a specific condition given the actions adopted by all the other players.  Under this game formulation, we introduce the concept of satisfaction equilibrium.

\subsection{Games in Satisfaction Form}
In general, a game in satisfaction-form can be described by the following triplet
\begin{equation}
\GameSF.
\end{equation}
Here, the set $\mathcal{K} = \lbrace 1, \ldots, K \rbrace$ represents the set of players and the set $\mathcal{A}_{k} = \lbrace A_{k}^{(1)}, \ldots, A_{k}^{(N_k)}\rbrace$ represents the set of $N_k$ actions available for transmitter $k$.
An action profile is a vector $\bs{a} = \left(a_{1}, \ldots, a_{K} \right) \in \mathcal{A}$, where,
\begin{equation}
\mathcal{A} = \mathcal{A}_1 \times \ldots \times \mathcal{A}_K.
\end{equation}
In this analysis, the sets $\mathcal{K}$ and $\left\lbrace \mathcal{A}_k \right\rbrace_{k \in \mathcal{K}}$ are assumed finite and non-empty.
We denote by $\bs{a}_{-k} = \left(a_1, \ldots, a_{k-1},a_{k+1},\ldots, a_K\right)\in \mathcal{A}_{-k}$, where
\begin{equation}
\mathcal{A}_{-k} = \mathcal{A}_1 \times \ldots \times \mathcal{A}_{k-1} \times \mathcal{A}_{k+1} \times \ldots, \times\mathcal{A}_K,
\end{equation}
the vector obtained by dropping off the $k$-th component of the vector $\bs{a}$. With a slight abuse of notation, we write the vector $\bs{a}$ as $\left(a_k,\bs{a}_{-k}\right)$, in order to emphasize its $k$-th component.
The correspondence $f_k: \mathcal{A}_{-k} \rightarrow 2^{\mathcal{A}_k}$ determines the set of actions of player $k$ which allows its satisfaction given the actions played by all the other players. Here, the notation $2^{\mathcal{A}_k}$ refers to the set of all possible subsets of the set $\mathcal{A}_k$, including $\mathcal{A}_k$. 
Note that $2^{\mathcal{A}_k}$  includes also the empty set, which models the case when one player ends up without  an action that allows the satisfaction of its individual constraints given the other players' actions. Often, this is a strong mathematical constraint and thus, in some sections of this paper, we assume that none of the correspondences $f_k$ maps the empty set.

\noindent
In general, an important outcome of a game in satisfaction form is the one where all players are satisfied. We refer to this game outcome as satisfaction equilibrium (SE).
\begin{definition}[Satisfaction Equilibrium in PS \cite{Perlaza-Globecom-2010}] \label{DefSE} \emph{An action profile $\bs{a}^+$ is an equilibrium for the game $\GameSF$ if
\begin{equation}\label{EqSE}
\forall k \in \mathcal{K}, \quad a_k^+ \in f_{k}\left(\bs{a}_{-k}^+\right).
\end{equation}
}
\end{definition}

\noindent
 Note that under this formulation, the outcome where all players are satisfied is naturally an equilibrium. Here, since the aim of each player is to be satisfied, none of them has a particular interest on changing its current action once they are at the SE. 
An important remark here is that, players are assumed to be careless of whether other players can satisfy or not their individual constraints. An interesting analysis of the impact of this assumption in the definition of equilibrium can be found in \cite{Altman-Solan-2009}.

\noindent
In this context, when DCSN are modeled using the satisfaction form,
 radio devices are indifferent to the fact that there might exist another transmit configuration with which a higher performance e.g., transmission rate, can be obtained. Here, as long as each radio device 
 is able to satisfy its individual QoS conditions, it has no incentive to unilaterally change its  current  transmit/receive configuration.

\noindent
In the following example, we show how a given decentralized self-configuring network can be modeled by a game in satisfaction form.
\begin{example}\label{ExSE}\emph{
Consider a decentralized and self-configuring network made of a set $\mathcal{K} = \lbrace 1, \ldots, K \rbrace$ of transmitter-receiver pairs.  For all $k \in \mathcal{K}$, let $\mathcal{A}_k$ be the set of transmit configurations available for transmitter $k$ and let the function $u_k: \mathcal{A}_1 \times \ldots \times \mathcal{A}_{K} \rightarrow \mathds{R}$ denote its (Shannon) transmission rate. Transmitter $k$ must guarantee  a data rate higher than $\Gamma_k$ bps. Hence, the set of configurations it must adopt, given the configurations $\bs{a}_{-k}$ of all the other transmitters, is determined by the correspondence $f_k: \mathcal{A}_{-k} \rightarrow 2^{\mathcal{A}_k}$, which we define as follows,
\begin{equation}\label{EqfkExample1}
f_k\left( \bs{a}_{-k} \right) =  \left\lbrace a_k \in \mathcal{A}_k: u_k\left( a_k, \bs{a}_{-k} \right) \geqslant \Gamma_k \right\rbrace.
\end{equation}
Thus, the behavior of this network can be modeled by the game $\GameSF$ in satisfaction form.
}
\end{example}

\noindent
In Sec. \ref{SecSEandOtherEquilibria}, we use this example to highlight the differences between the satisfaction form and other game formulations.
In the following, we describe the extension in mixed strategies of the game in satisfaction form.

\subsection{Extension in Mixed Strategies of the Satisfaction Form}

The concept of mixed strategies was introduced by Borel in \cite{Borel-1921}.  A mixed strategy of player $k$ is a probability distribution over the set of actions $\mathcal{A}_k$.
We denote the set of all possible probability distributions over the set $\mathcal{A}_k$ by $\triangle\left( \mathcal{A}_k \right)$, i.e., the unit simplex over the elements of $\mathcal{A}_k$.
We denote by $\bs{\pi}_k = \left(\pi_{k,A_k^{(1)}}, \ldots, \pi_{k,A_k^{(N_k)}} \right)$ the probability distribution (mixed strategy) chosen by player $k$. Here, for all $n_k \in \lbrace 1, \ldots, N_k\rbrace$, $\pi_{k, A_{k}^{(n_k)}}$ represents the probability that player $k$ plays action $A_{k}^{(n_k)} \in \mathcal{A}_k$.

\noindent
Following this notation, we denote by $\GameSFmixed$ the extension in mixed strategies of the game $\GameSF$, where the correspondence
\begin{equation}
\bar{f}_k: \ds\prod_{j\in\mathcal{K}\setminus\lbrace k \rbrace} \triangle\left(\mathcal{A}_j\right) \rightarrow 2^{\triangle\left(\mathcal{A}_{k}\right)},
\end{equation}
determines the set of all possible probability distributions that allow player $k$ to always choose an action which satisfies its individual conditions, that is,
\begin{eqnarray}
\label{EqFmixed}
\bar{f}_k\left( \bs{\pi}_{-k} \right) & = & \left\lbrace \bs{\pi}_k \in \triangle\left( \mathcal{A}_k \right): \right.\\
\nonumber
& & \left. \pr{a_k \in f_k\left( \bs{a}_{-k} \right)}{\left(\bs{\pi}_k, \bs{\pi}_{-k}\right)} = 1\right\rbrace.
\end{eqnarray}

In this context, we define the SE as follows.
\begin{definition}[Satisfaction Equilibrium in MS]\label{DefSEmixed}\emph{
The mixed strategy profile $\bs{\pi}^* \in \triangle\left(\mathcal{A}_1\right)\times\ldots\times\triangle\left(\mathcal{A}_K\right)$ is a SE of the game $\GameSFmixed$, if for all $k \in \mathcal{K}$,
\begin{equation}
\bs{\pi}_k^* \in \bar{f}_k\left( \bs{\pi}_{-k}^*\right).
\end{equation}
}
\end{definition}
From Def. \ref{DefSEmixed} and \eqref{EqFmixed}, it can be implied that if $\bs{\pi}^* \in \triangle\left(\mathcal{A}_1\right)\times\ldots\times\triangle\left(\mathcal{A}_K\right)$ is a SE, then the following holds, for all $k \in \mathcal{K}$,
\begin{equation}
\pr{a_k \in f_k\left(\bs{a}_{-k}\right)}{\left(\bs{\pi}_k^*,\bs{\pi}_{-k}^*\right)} = 1.
\end{equation}
Note that it can be stated that the set of equilibria of the game $\GameSF$ is a subset of the set of equilibria of the mixed extension $\GameSFmixed$, if we establish an injection from the action set $\mathcal{A}_k$ to the canonical basis of the space of $N_k$-dimensional vectors $\mathds{R}^{N_k}$. For instance, let the $n_k$-th action of player $k$, i.e., $A_{k}^{(n_k)}$, be associated with the unitary vector $\bs{e}_{n_k}^{(N_k)}  = \left(e_{n_k,1}^{(N_k)}, \ldots, e_{n_k,N_k}^{(N_k)} \right) \in \mathds{R}^{(N_k)}$, where, all the components of the vector $\bs{e}_{n_k}^{(N_k)}$ are zero except its $n_k$-th component, which is unitary. The vector $\bs{e}_{n_k}^{(N_k)}$ represents a degenerated probability distribution, where the action $A_k^{(n_k)}$ is deterministically chosen. Using this argument, it becomes clear that every SE of the game $\GameSF$ is also a SE in the game $\GameSFmixed$.

\noindent
As we shall see in the next section,  games in satisfaction form might not necessarily have a SE neither in pure strategies nor mixed strategies. Thus, in the following we present a less restrictive notion of equilibrium to which we refer to as epsilon-satisfaction equilibrium ($\epsilon$-SE).

\subsection{Epsilon-Satisfaction Equilibrium}

At the SE of the game $\GameSFmixed$, players choose their actions following a probability distribution such that only action profiles that allow all players to simultaneously satisfy their individual conditions are played with positive probability.
This interpretation leads immediately to the conclusion that if it does not exist at least one action profile that allows all players to be simultaneously satisfied, then, it does not exists any SE in the game  $\GameSFmixed$. However, under certain conditions, it is always possible to build mixed strategies that allow players to be satisfied with a probability which is close to $1$ , i.e., $1-\epsilon$, for a sufficiently small $\epsilon > 0 $.

\begin{definition}[Epsilon-Satisfaction Equilibrium]\label{DefepsilonSEmixed}\emph{
Let $\epsilon \in \left] 0, 1\right]$. The mixed strategy profile $\bs{\pi}^* \in \triangle\left(\mathcal{A}_1\right)\times\ldots\times\triangle\left(\mathcal{A}_K\right)$ is an epsilon-satisfaction equilibrium ($\epsilon$-SE) of the game $\GameSFmixed$, if for all $k \in \mathcal{K}$, it follows that
\begin{equation}
\bs{\pi}_{k}^* \in \bar{\bar{f}}_k\left( \bs{\pi}_{-k}^* \right),
\end{equation}
where
%
\begin{eqnarray}
\label{EqFepsilon}
\bar{\bar{f}}_k\left( \bs{\pi}_{-k}^* \right) & = & \left\lbrace \bs{\pi}_k \in \triangle\left( \mathcal{A}_k \right): \right. \\
\nonumber
& & \left. \pr{a_k \in f_k\left( \bs{a}_{-k} \right)}{\left(\bs{\pi}_k, \bs{\pi}_{-k}^*\right)} \geqslant 1-\epsilon \right\rbrace.
\end{eqnarray}
}
\end{definition}
From Def. \ref{DefepsilonSEmixed}, it can be implied that if the mixed strategy profile $\bs{\pi}^*$ is an $\epsilon$-SE, it holds that,
\begin{equation}
\pr{a_k \in f_k\left( \bs{a}_{-k} \right)}{\left(\bs{\pi}_k^*, \bs{\pi}_{-k}^*\right)} \geqslant 1 - \epsilon.
\end{equation}
That is, players are unsatisfied with probability $\epsilon$. The relevance of the $\epsilon$-SE is that it models the fact that players can be tolerant to a non-satisfaction level.  At a given $\epsilon$-SE, none of the players is interested in  changing its mixed strategy profile as long as it is satisfied with a probability higher than certain threshold $1 - \epsilon$.
As we shall see, a game might not possess a SE neither in pure nor mixed strategies, but it might possess an $\epsilon$-SE.
 
 \noindent
A thorough  analysis on the existence and uniqueness of the SE in pure strategies and mixed strategies is presented in the next section. Similarly, the conditions  for the existence of an $\epsilon$-SE are also discussed.

\section{Existence and Uniqueness of the Satisfaction Equilibrium}\label{SecExistenceAndUniqueness}

In this section, we study the existence and uniqueness of a satisfaction equilibrium in games in satisfaction form and its corresponding extension in mixed strategies. Particular attention is given to the existence of $\epsilon$-SE in the case where it does not exist at least one SE neither  in pure nor in mixed strategies.

\subsection{Existence of SE in Pure Strategies}\label{SecExistenceSEinPS}

In order to study the existence of a SE in the game $\GameSF$, let the correspondence $F: \mathcal{A} \rightarrow  2^{\mathcal{A}}$ be defined as follows:
\begin{equation}
\label{EqBigF}
F(\bs{a}) = f_1\left( \bs{a}_{-1}\right) \times \ldots \times f_K\left( \bs{a}_{-K}\right).
\end{equation}
Then, a SE exists if and only if
\begin{equation}\label{EqExistence}
\exists \bs{a} \in \mathcal{A}:  \quad	 \bs{a} \in F(\bs{a}).
\end{equation}
This formulation allows us to use existing fixed point (FP) theorems to provide sufficient conditions for the existence of the SE. For instance, one can rely on the fixed point theorem of Knaster and Tarski \cite{Knaster-Tarski-1928} to state the following theorem.
\begin{theorem}[Existence of SE in finite games] \label{TheoExistence2}\emph{
Consider the game $\GameSF$ and  let the set $\mathcal{A}$ have a binary relation denoted by $\preceq$. Let also \\
  $(i)$ $V = \left< \mathcal{A}, \preceq \right>$ be a complete lattice;\\
  $(ii)$ $F\left( \bs{a} \right)$ be non-empty for all $\bs{a}\in \mathcal{A}$;\\
  $(iii)$ the correspondence $F$ in \eqref{EqBigF} satisfies that $\forall \left( \bs{a}, \bs{a}' \right) \in \mathcal{A}^2$, such that $\bs{a} \preceq \bs{a}'$, it holds that
  \begin{equation}
\label{EqOrderedCorrespondance}
\forall \left( \bs{b}, \bs{b}' \right) \in F\left(\bs{a}\right)\times F\left(\bs{a}'\right), \quad \bs{b} \preceq \bs{b}'.
\end{equation}
Then the game has at least one SE in pure strategies.}
 \end{theorem}
Note that  theorem  \ref{TheoExistence2} requires that for all $\bs{a} \in \mathcal{A}$, the set $F\left( \bs{a} \right)$ is non-empty, i.e.,
 \begin{equation}
\label{EqConditionForExistence}
\forall k \in \mathcal{K} \text{ and }  \forall \bs{a}_{-k} \in \mathcal{A}_{-k}, \, \exists
a_k \in \mathcal{A}_k: \, a_k \in f_{k}\left( \bs{a}_{-k}\right).
\end{equation}
In some cases, this condition might appear restrictive. However, in the general context of wireless communications, when a radio device might not provide its QoS, the default action is simply going into standby. This might imply the existence of an action ``do nothing (DN)'' which might appear to avoid the emptiness of $f_k(\bs{a}_{-k})$, whenever it is required. Modeling the existence of the DN action strongly depends on the scenario. For instance, in the case of power allocation games, such action can be the null vector, that is, zero transmit power.

\noindent
In the following, we study the existence of an equilibrium in mixed strategies and an $\epsilon$-SE, which appear to be less restrictive.

\subsection{Existence of the SE in Mixed Strategies}

As in the case of pure strategies, the condition for the existence of a SE in the mixed extension $\GameSFmixed$ boils down to the study of a fixed point inclusion.  Let the correspondence $\bar{F}: \triangle\left(\mathcal{A}_1 \right) \times \ldots \times \triangle\left(\mathcal{A}_K \right) \rightarrow  2^{\triangle\left(\mathcal{A}_1 \right) \times \ldots \times \triangle\left(\mathcal{A}_K \right)}$ be defined as follows:
\begin{equation}
\label{EqBigBarF}
\bar{F} (\bs{\pi}) = \bar{f}_1\left( \bs{\pi}_{-1}\right) \times \ldots \times \bar{f}_K\left( \bs{\pi}_{-K}\right).
\end{equation}
Then, a SE exists if and only if
\begin{equation}\label{EqExistenceMixed}
\exists \bs{\pi} \in \triangle\left(\mathcal{A}_1 \right) \times \ldots \times \triangle\left(\mathcal{A}_K \right):  \quad	 \bs{\pi} \in \bar{F}(\bs{\pi}).
\end{equation}
Thus, all the results of fixed point theory \cite{Border-1985}, in the case of the compact and convex sets, are valid for the study of the existence of the SE in the game $\GameSFmixed$.
Nonetheless, some results are immediate from Def. \ref{DefSEmixed}. For instance, note that if a game in satisfaction form does not have a SE in pure strategies, then, it does not have a SE in mixed strategies neither. This is basically due to the fact that players mix only the actions that guarantee  their satisfaction with unitary probability.
That is, player $k$ mixes a subset of its actions $\mathcal{A}_k' \subseteq \mathcal{A}_k$, i.e., $\forall a_k \in \mathcal{A}_k'$, $\pi_{k,a_k} > 0$, only if the following condition holds, $\forall a_k^* \in \mathcal{A}_k'$,
\begin{eqnarray}
\pr{a_k^* \in f_{k}\left( \bs{a}_{-k}\right)}{\bs{\pi}_{-k}} &= & 1.
\end{eqnarray}
This implies that player $k$ assigns a strictly positive probability to more than one action, i.e., it plays strictly mixed strategies, only if such a set of actions guarantees its satisfaction for all the action profiles $\bs{a}_{-k} \in \mathcal{A}_{-k}$, which are played with non-zero probability. This reasoning might imply that, there might exist several SE in pure strategies but no SE in strictly mixed strategies in the game $\GameSFmixed$.

\subsection{Existence of the Epsilon-Satisfaction Equilibrium}

As shown in the previous subsection, the existence of at least one SE  in the extension in mixed strategies of the game $\GameSF$ remains very strict. Indeed, the game has a SE in mixed strategies if and only if it has a SE in pure strategies. On the contrary, the existence of at least one $\epsilon$-SE is less strict and it does not require the existence of a pure SE. A sufficient and necessary condition for the existence of at least one   $\epsilon$-SE is the following.
\begin{proposition}\emph{
Let $\GameSF$ be a finite game in satisfaction form. Then, if the following condition holds,
\begin{equation}\label{EqEpsilonEQcondition}
\forall k \in \mathcal{K}, \, \exists \bs{a} \in \mathcal{A}: \quad a_k \in f_k\left( \bs{a}_{-k} \right),
\end{equation}
there always exists a strategy profile $\bs{\pi}^* \in \triangle\left(\mathcal{A}_1\right) \times \ldots \times \triangle\left(\mathcal{A}_K\right)$ and an $1 > \epsilon >  0$, such that, $\bs{\pi}^*$ is an $\epsilon$-SE.
}
\end{proposition}
\begin{proof}\emph{
Assume that the condition \eqref{EqEpsilonEQcondition} holds. Then, for all $j \in \mathcal{K}$, it holds that the set
\begin{equation}
\mathcal{A}_j^* = \left\lbrace  \bs{a} \in \mathcal{A} : a_j \in f_j\left( \bs{a}_{-j} \right)  \right\rbrace
\end{equation}
is non-empty. Denote by $\bs{a}_j^* = \left( a_{j,1}^*, \ldots,  a_{j,K}^*\right)$ a particular element of the set $\mathcal{A}_j^*$.
Any mixed strategy $\bs{\pi}^{+} \in \triangle\left( \mathcal{A}_1\right)\times\ldots\times\triangle\left( \mathcal{A}_K\right)$, such that
\begin{equation}
\forall (j,k) \in \mathcal{K}^2, \quad \pi_{k,a_{j,k}^*}^+ > 0
\end{equation}
guarantees that $\forall j \in \mathcal{K}$, the action $\bs{a}_{j}^*$ is played with non-zero probability, thus, $\forall k \in \mathcal{K}$,
\begin{eqnarray}\nonumber
\pr{a_k \in f_k\left( \bs{a}_{-k}\right)}{\bs{\pi}^+} & = & \ds\sum_{\bs{a}\in \mathcal{A}} \ind{a_k \in f_k(\bs{a}_{-k})}\ds\prod_{j = 1}^K \pi_{j,a_j}^+ \\
& & = \epsilon_k,
\end{eqnarray}
where $1 \geqslant \epsilon_k \geqslant \ds\prod_{j = 1}^K \pi_{j,a_{k,j}^*}^+ > 0$,
which proves the existence of a mixed strategy profile such that, $\forall k \in \mathcal{K}$,
\begin{equation}
\pr{a_k \in f_k\left( \bs{a}_{-k}\right)}{\bs{\pi}^+} \geqslant 1 - \epsilon,
\end{equation}
where, $\epsilon = 1 - \ds\min_{j \in \mathcal{K}} \epsilon_j$, which completes the proof.
}
\end{proof}
Note that for all $k \in \mathcal{K}$, the condition \eqref{EqEpsilonEQcondition}  only requires the existence of at least one action profile where player $k$ is satisfied, which is less restrictive than conditions \eqref{EqExistence} and \eqref{EqExistenceMixed}. Note also that, as long as \eqref{EqEpsilonEQcondition} holds,  a simple uniform distribution over each individual set of actions $\mathcal{A}_k$ is an $\epsilon$-SE. Indeed, under the assumption of a uniform probability distribution, the following lower bound can be identified $\epsilon \geqslant 1 - \ds\prod_{j =1}^K \frac{1}{N_j}$. Indeed, strict equality is observed in a  game with no SE in pure strategies and satisfying that $\forall k \in \mathcal{K}$, there exists a unique action profile $\bs{a}^{(k)}$ such that $a_k^{(n_k)} \in f_k\left( a_{-k}^{(n_k)} \right)$ and $\forall (j,k) \in \mathcal{K}^{2}$, $\bs{a}^{(j)} \neq\bs{a}^{(k)}$.
 
\subsection{Multiplicity of the SE and Equilibrium Selection}\label{SecESE}
In general, it is difficult to provide the conditions to observe a unique SE for a general set of correspondences $\lbrace f_{k} \rbrace_{k \in \mathcal{K}}$. As we shall see in Sec. \ref{SecApplications}, the set of SE is often not unique in games modeling decentralized self-configuring wireless networks,  and thus, an equilibrium selection process might be required.

\noindent
We start our analysis pointing the fact that using a higher transmit power level or using a more complex modulation scheme (e.g., in the sense of the size of the constellation) might require a higher energy consumption and thus, reduce the battery life time of the transmitters. In this scenario, one might imply that radio devices are interested in satisfying their required QoS with the lowest effort. Here, we can express the effort, for instance, in terms of energy consumption or signal processing complexity.

\noindent
In the following, our interest is to select an action profile which allows all the players to be satisfied with the lowest effort. We refer to this action profile as efficient satisfaction equilibrium (ESE).

\noindent
The game where each player aims to satisfy its QoS with the minimum effort can be formulated as a game in normal form with constrained set of actions,
\begin{equation}\label{EqGameTilde}%
\GameESF.
\end{equation}
Here, for all  $k \in \mathcal{K}$, the cost function $c_k : \mathcal{A}_k \rightarrow \left[ 0, 1 \right]$ satisfies the following condition, $\forall (a_k, a_k') \in \mathcal{A}_k^2$, it holds that
\begin{equation}
\label{EqFunctionC}
c_k\left( a_k \right) < c_k\left( a_k' \right),
\end{equation}
if and only if, $a_k$ requires a lower effort than action $a_k'$ when it is played by player $k$.
For instance, if transmitters are choosing the transmit power level, one can consider the cost function $c_k$ as the identity function, that is, the transmitters are concerned by the amount of power they use on their transmissions.

\noindent
Following this reasoning, the set of ESE of the game $\GameSF$, with respect to the cost functions $c_k$, corresponds to  the set of GNE of the game $\GameESF$.

\noindent
It is important to note that in the game $\GameESF$, the competitive interaction between all players is not modeled by the cost functions $\lbrace c_k \rbrace_{k \in \mathcal{K}}$. For instance, the cost function of player $k$, $c_k$, depends only on its chosen action $a_k$. In this game formulation $\tilde{\mathcal{G}}$, the interaction between players is modeled by the correspondence $f_k$, which is defined over the set of action profiles $\mathcal{A}_{-k}$.

\noindent
An important remark here is that if all players assign the same cost (or effort) to all their actions, then the sets of ESE and the set of SE are identical. This implies that the interest of the formulation $\tilde{\mathcal{G}}$ is precisely that players can differentiate the effort of playing one action or another in order to select one (satisfaction) equilibrium among all the existing equilibria of the game $\gameSF$. Thus, the existence and uniqueness of this efficient SE plays an important role in the equilibrium selection. We analyze these two properties in the sequel of this section.
\subsubsection{Existence of an ESE}
\noindent
The main results on the existence of an ESE are based on the fact that the game $\GameESF$ is a potential game with constrained actions sets (see Appendix \ref{AppPGCSA}), as shown in the following proposition.

\begin{proposition}\label{PropGtildeisECPG} \emph{
The game $\GameESF$  is an exact constrained potential game, with potential function $\phi: \mathcal{A} \rightarrow \mathds{R}$, such that, $\forall \bs{a} \in \mathcal{A}$,
\begin{equation}\label{EqPotential}
\phi\left( \bs{a} \right) = \ds\sum_{k = 1}^{K} c_{k}\left( a_k \right).
\end{equation}
}
\end{proposition}
This simple observation allows us to obtain immediate results on the existence of the ESE:
\begin{proposition}[Existence of the ESE]\label{PropExistenceESE} \emph{
The  game $\GameESF$ \eqref{EqGameTilde} has at least one efficient SE if the graphs of the correspondences $f_1, \ldots, f_K$ are non-empty and identical.
}
\end{proposition}
The proof of Prop. \ref{PropExistenceESE} follows from Th. \ref{TheoExistencePureStrategiesECPG}  in Appendix \ref{AppPGCSA}. Other results, on the existence of an ESE and more importantly on the convergence of the best response dynamics \cite{Rose-Perlaza-CommMag2011} to achieve an ESE are immediately obtained by exploiting properties such as supermodularity or submodularity of the potential function (See \cite{Altman-Altman-2003}).

\subsubsection{Multiplicity of an ESE}

One the desirable features of the ESE is to be unique. In particular, this will allow engineers to estimate the operating point of decentralized networks. However, from existing literature in constrained games \cite{Altman-Altman-2003}, it can be stated that uniqueness is observed only under certain conditions. Thus, in Appendix \ref{AppMultiplicityESE}, we provide a general method to analyze the number of ESE of a given game $\gameESF$. This method is basically a tool for the network analysis, it  does not pretend to be an algorithm which can be directly implemented in networks.

\noindent
In the next section, we compare the SE notion with existing equilibrium concepts such as NE and GNE.

\section{Satisfaction Equilibrium and other Equilibrium Concepts}\label{SecSEandOtherEquilibria}

In the following, we highlight the main differences between the SE and other equilibrium notions such as NE and GNE. However, before we start, we point out the differences between the normal form and the satisfaction form formulations.
\subsection{Games in Normal Form and Satisfaction Form }
The main difference between the normal form $\GameNF$ and the satisfaction form $\GameSF$ is that the former defines a preference order \cite{Fishburn-1970}, which can be modeled by a family of utility functions $u_1, \ldots, u_K$ in the sense of von Neumann - Morgenstern \cite{Morgestein-Neumann-44}, i.e., given an action profile $\bs{a}_{-k} \in \mathcal{A}_{-k}$, player $k$ can rank any pair of its actions $(a_k,a_k') \in \mathcal{A}_k^2$ such that either $u_k(a_k,\bs{a}_{-k}) < u_k(a_k',\bs{a}_{-k})$,  $u_k(a_k,\bs{a}_{-k}) = u_k(a_k',\bs{a}_{-k})$ or $u_k(a_k,\bs{a}_{-k}) > u_k(a_k',\bs{a}_{-k})$. In the latter, player $k$ determines only whether an action satisfies its individual conditions or not, i.e.,  $a_k \in f_k\left( \bs{a}_{-k}\right)$ or $a_k \notin f_k\left( \bs{a}_{-k}\right)$, respectively. Hence, the satisfaction form does not require the existence of a utility function \cite{Osborne-Rubinstein-Book, Lasaulce-Tembine-Book-2011} but a satisfaction correspondence $f_k$, which can be seen as an additional degree of freedom of this formulation.

\noindent
An interesting observation is that, by assuming a utility function of the form $u_k: \mathcal{A} \rightarrow \lbrace 0,1 \rbrace$ a game in normal form can also be used to describe a situation when players are interested only on the satisfaction of their individual constraints. Nonetheless, the notions of Nash equilibrium (NE) and SE do not necessarily correspond to each other. The NE in pure strategies in the context of games in normal form \cite{TiroleBook1991} can be defined as follows.
\begin{definition}[Nash Equilibrium in PS\cite{Nash-1950}]\label{DefNE} \emph{Consider a game $\GameNF$ in normal form. An action profile $\boldsymbol{a} \in \mathcal{A}$ is an NE in pure strategies if it satisfies, for all $k \in \mathcal{K}$ and for all $\bs{a}'_k \in \mathcal{A}_k$,
\begin{equation}
u_k(a_k,\bs{a}_{-k}) \geqslant  u_k(a'_k,\bs{a}_{-k}).
\end{equation}}
\end{definition}
Now, consider the Ex. \ref{ExSE} and define the game in satisfaction form $\GameSF$, with $f_k$ defined by \eqref{EqfkExample1} and a game in normal form  $\GameNF$, with
$u_k: \mathcal{A}_1 \times \ldots \times  \mathcal{A}_K \rightarrow \lbrace 0,1\rbrace$ 
defined as follows: 
\begin{equation}
\label{EqVtility}
u_k\left( a_k, \bs{a}_{-k}\right) = \ind{a_k \in f_{k}\left( \bs{a}_{-k}\right)}.
\end{equation}
Now, we compare both the set of SE $\mathcal{A}_{\mathrm{SE}}$ of the game  $\GameSF$ and the set of NE $\mathcal{A}_{\mathrm{NE}}$ of the game $\GameNF$. Note that from Def. \ref{DefNE} and Def. \ref{DefSE}, it can be immediately implied that any SE of the game $\gameSF$ is an NE of the game $\gameNFmodif$. This is basically, because at the SE, all players obtain a unitary utility, and since the range of the utility function is binary $\lbrace 0, 1 \rbrace$, no other action is able to give a  higher utility. The converse is not true, that is, an NE of the game $\GameNF$ is not necessarily a SE of the game $\GameSF$. Consider for instance the game realization ($K = 2$, $N_1 = N_2 = 2$) in Fig. \ref{FigMatrixExample}. Note that therein, the game $\GameNF$ has $2$ NE in pure strategies, which are the action profiles $(A_1^{(2)}, A_2^{(1)})$ and $(A_1^{(1)},A_2^{(2)})$, while the game $\GameSF$ has only one SE, which is the action profile  $(A_1^{(1)},A_2^{(2)})$.
\begin{figure}[h]
\begin{center}
$\begin{array}{|c|c|c|}\hline
\scriptstyle P_1 \backslash P_2 & A_2^{(1)} &  A_2^{(2)} \\ \hline
A_1^{(1)}  & (0,0) &  (1,1) \\ \hline
A_1^{(2)} & (1,0) &  (0,0) \\ \hline
\end{array}$
\end{center}
\caption{Game in normal form $\GameNF$, with $\mathcal{K}=\lbrace 1, 2 \rbrace$,  $\mathcal{A}_k = \lbrace A_k^{(1)}, A_k^{(2)} \rbrace$, for all $k \in \mathcal{K}$.  Player $1$ chooses rows and player $2$ chooses columns. In a pair $(v_1,v_2) \in \lbrace 0,1 \rbrace^{2}$, $v_1$ and $v_2$ are the utilities obtained by player $1$ and $2$, respectively.}
\label{FigMatrixExample}
\end{figure}
This simple example shows that, the formulation in normal-form might lead to equilibria where not all the players are satisfied even when the joint satisfaction of all players is feasible. This shows that games in normal-form do not properly model the case where players are interested only in the satisfaction of individual conditions.  We conclude the comparison between the games $\GameNF$ and $\GameSF$, by establishing the following condition between their sets of equilibria,
\begin{equation}\label{EqSetOfEinSEandNE}
\mathcal{A}_{\mathrm{SE}} \subseteq \mathcal{A}_{\mathrm{NE}} \subseteq \mathcal{A}.
\end{equation}
This confirms the intuition that the notion of SE is more restrictive than the notion of NE, that is, an SE in the game $\gameSF$ is a Pareto optimal NE in the game $\gameNF$, that is, an action profile where all players are satisfied.

\subsection{Satisfaction Equilibrium  and Nash Equilibrium}
The definition of NE (Def. \ref{DefNE}) can be obtained from the definition of SE (Def. \ref{DefSE}) by assuming that, for all $k \in \mathcal{K}$, the satisfaction correspondence $f_k$ is defined as follows,
\begin{equation}\label{EqFBestResponse}
f_k \left( \bs{a}_{-k} \right) =  \ds\arg\max_{a_k^* \in \mathcal{A}_k} u_k\left( a_k^*, \bs{a}_{-k}\right).
\end{equation}
The satisfaction correspondence $f_k$ as defined in \eqref{EqFBestResponse} is known in the game theoretic literature as the best response correspondence \cite{TiroleBook1991}. Then, under this formulation, the set of SE of the game in satisfaction form $\GameSF$ is identical to the set of NE of the game in normal form $\GameNF$. This reasoning might lead us to think that the satisfaction form  as well as the notion of SE are generalizations of the classical normal form and the notion of Nash equilibrium \cite{Nash-1950}, respectively.

\subsection{Satisfaction Equilibrium and Generalized Nash Equilibrium}
The GNE in pure strategies (PS) in games in normal form with constrained set of actions, as introduced by Debreu in \cite{Debreu-1952} and later by Rosen in \cite{Rosen-65}, can be defined as follows.
 \begin{definition}[Generalized NE in PS \cite{Debreu-1952}]\label{DefGNE} \emph{An action profile $\boldsymbol{a}^* \in \mathcal{A}$ is a generalized Nash equilibrium (GNE) of the game $\GameSE$ if and only if
\begin{equation}\nonumber
\forall k \in \mathcal{K}, \quad a_k^* \in f_{k}\left(\bs{a}_{-k}^*\right) \mbox{ and }	
\end{equation}
\begin{equation}\nonumber
	\forall a_k \in f_k\left(\bs{a}_{-k}^*\right), \; u_k(a_k^*,\bs{a}_{-k}^*) \geqslant u_k(a_k,\bs{a}_{-k}^*).
\end{equation}
}
\end{definition}
Note that the definition of SE (Def. \ref{DefSE}) can be obtained from the definition of GNE (Def. \ref{DefGNE}) by assuming the following condition, $\forall k \in \mathcal{K}$ and $\forall \bs{a} \in \mathcal{A}$
\begin{equation}\label{EqSEandGNE}
 \quad u_k(a_k,\bs{a}_{-k}) = c, \quad  \mbox{with $c \in \mathds{R}_+$}.
\end{equation}
Under assumption \eqref{EqSEandGNE}, the set of GNE of the game in normal form with constrained set of actions $\GameSE$ and the set of SE of the game  $\GameSF$ in satisfaction form  are identical. This observation does not necessarily imply that the satisfaction form  as well as the notion of SE are particular cases of the classical normal form with constrained set of actions and the notion of GNE \cite{Debreu-1952}, respectively. Indeed, there exist fundamental differences: $(i)$ in the game $\gameSE$, the set of available actions for player $k$ are determined by the complementary vector $\bs{a}_{-k}$. On the contrary, in the game $\gameSF$, the set of available actions of player $k$ is always the set $\mathcal{A}_k$. $(ii)$ In the game $\gameSE$, a rational player $k$ determines the action to play $a_k$ by following two different steps. First, it determines the set of available actions $f_{k}(\bs{a}_{-k})$ and second, it determine the actions $a_k \in f_{k}(\bs{a}_{-k})$ that maximize  $u_k$. In contrast, in the game $\gameSF$, player $k$ does not require any optimization capability. $(iii)$ In the game $\gameSE$, the interpretation of $f_k(\bs{a}_{-k}) = \emptyset$ is that player $k$ cannot play since none of its actions is available given the actions of the others $\bs{a}_{-k}$. On the contrary, in the game $\gameSF$,  the interpretation of $f_k(\bs{a}_{-k}) =\emptyset$ is that player $k$ can take any of its actions $a_k \in \mathcal{A}_k$, but none of them achieves satisfaction. This difference might appear subtile but it makes a strong difference when the equilibrium must be learnt dynamically \cite{Rose-Perlaza-CommMag2011}. Indeed, in an eventual exploration phase of a learning algorithm, at each stage, player $k$ always has a non-empty set of actions to test in the game $\gameSF$ regardless of the actions of all the other players. Clearly, this is not the case in the game $\gameSE$, which constraints the learning process.  
 
\noindent
In the following, we compare the set of equilibria of both games $\GameSE$ and $\GameSF$, for a general definition of the utility functions $u_k$, for all $k \in \mathcal{K}$.
Let the sets of GNE of the game $\gameSE$ and the set of SE of the game $\gameSF$ be denoted by $\mathcal{A}_{\mathrm{GNE}}$ and $\mathcal{A}_{\mathrm{SE}}$, respectively. Now, note that from Def. \ref{DefGNE} and Def. \ref{DefSE}, it follows  that any GNE in $\gameSE$ is a SE in $\gameSF$, i.e.,
\begin{equation}
\mathcal{A}_{\mathrm{GNE}} \subseteq \mathcal{A}_{\mathrm{SE}} \subseteq  \mathcal{A}.
	\label{EqSetOfEinNFCAS}
\end{equation}
The strict equality $\mathcal{A}_{\mathrm{GNE}} = \mathcal{A}_{\mathrm{SE}}$ is achieved when the functions $u_k$ are chosen following \eqref{EqSEandGNE}.
The condition in \eqref{EqSetOfEinNFCAS} verifies the intuition that the notion of SE in games in satisfaction form, is less restrictive than the notion of GNE in games in normal form with constrained action sets. Note also that from Def. \ref{DefSE}, it might be implied that several SE might exist, while no GNE necessarily exists. This is basically due to the fact that the existence of a GNE depends on both the functions $u_k$ and $f_k$, while the existence of a SE depends uniquely on the correspondences $f_k$, with $k \in \mathcal{K}$.
Conversely, the existence of a GNE implies the existence of at least one SE.

\noindent
In the following, we focus on designing behavioral rules for the radio devices in order to let them to learn one satisfaction equilibrium in decentralized self-configuring networks.
\section{Learning Satisfaction Equilibrium}\label{SecLearningSE}

In this section, we study a behavioral rule that allows radio devices to learn a satisfaction equilibrium in a fully decentralized fashion. Here, the underlying assumption is that players do not need to observe the value of its achieved utility, i.e., transmission rate, energy efficiency, etc., but only to know whether they are satisfied or not at each stage of the learning process. This implies only a $1$-bit length message exchange between the corresponding transmitter-receiver pairs. In the following, we formulate the corresponding  learning problem and later, we introduce the behavioral rules that allow players to learn the SE.

\subsection{The Learning Problem Formulation}\label{SecLearningProblemFormulation}
We describe the SE learning process in terms of elements of the game $\GameSF$ as follows.  Assume that time is divided in intervals and denote each interval by the index $n \in \mathds{N}$. Each interval ends when  each player  has played at most once. Denote the action taken by player $k$ at interval $n$ by $a_k(n)$. At each interval $n$, player $k$ observes whether it is satisfied or not, i.e., it observes a binary variable
\begin{equation}\label{EqTildeV}
\tilde{v}_k(n) = \ind{a_k(n) \in f_k\left( \bs{a}_{-k}(n)\right)}.
\end{equation}
 Our intention is to learn at least one  SE by letting the players to interact following particular behavioral rules. We say that players learn an equilibrium in pure strategies if, after a given finite number of time intervals, all players have chosen an action which achieves satisfaction, and thus, no other action update takes place. %
\subsection{Learning the SE in Pure Strategies}\label{SecLearningSEinPS}
Before we present the behavioral rule which allows players to achieve one of the equilibrium of the game $\GameSF$, we state the following hypothesis:\\

$(i)$ The game $\GameSF$ has at least one SE in pure strategies.\\

$(ii)$ For all $k \in \mathcal{K}$, it holds that $\forall \bs{a}_{-k} \in \mathcal{A}_{-k}$, the set $f_k\left( \bs{a}_{-k}\right)$ is not empty. \\

$(iii)$ The sets $\mathcal{K}$ and $\left\lbrace \mathcal{A}_k \right\rbrace_{k \in \mathcal{K}}$, are finite.\\

\noindent
The first hypothesis ensures that the SE learning problem is well-posed, i.e., radio devices are assigned a feasible task. The second hypothesis refers to the fact that, each radio device is  always able to find a transmit/receive configuration with which it can be considered satisfied given the transmit/receive configuration of all the other radio devices. This assumption might appear restrictive but it is not necessarily the case, see the discussion on the ``do nothing'' action in Sec. \ref{SecExistenceSEinPS}. 
The third hypothesis is considered in order to ensure that our algorithm is able to converge in finite time.

\noindent
Under the assumption that all hypothesis hold, each player chooses its own action as follows. The first action of player $k$, denoted by $a_k(0)$, is taken following an arbitrary probability distribution $\hat{\bs{\pi}}_k(0) \in \triangle\left(\mathcal{A}_k\right)$. Often, such a probability $\hat{\bs{\pi}}_k(0)$ is the uniform probability distribution.
At time interval $n > 0$,  player $k$ changes its action if and only if it is not satisfied, i.e, $\tilde{v}_k(n-1) = 0$. In this case, the next action is chosen following a probability distribution $\hat{\bs{\pi}}_k(n)$ to which we refer to as probability distribution of exploration.  If player $k$ is satisfied  i.e, $\tilde{v}_k(n-1) = 1$, then, it keeps playing the same action.
Hence, we can write that,
\begin{eqnarray}\label{EqActionAtIntervalN}
	a_k(n) =\left\lbrace
	\begin{array}{lcl}
	a_k(n-1) & \text{ if } & \tilde{v}_{k}(n-1) = 1  \\
	a_k(n) \sim \hat{\bs{\pi}}_k(n)  & \text{ if } & \tilde{v}_{k}(n-1) = 0
	\end{array}. 						
	\right.
\end{eqnarray}
The behavioral rule \eqref{EqActionAtIntervalN} is based on the proposal in \cite{Ross-07}. Here, note that a larger class of individual constraints can be considered due to the formulation of \eqref{EqTildeV}. Under this formulation, the room for optimization is on the design of $\hat{\bs{\pi}}_k(n)$ and its evolution over time. However, we left this issue out of the scope of the paper and no particular probability distribution is assumed.
We formalize the behavioral rule \eqref{EqActionAtIntervalN} in the Alg. \ref{Algo1}. Regarding the convergence of this behavioral rule, we provide the following proposition.
\begin{algorithm}
\begin{algorithmic}[1]
\REQUIRE { At each instant $n > 0$: $\tilde{v}_k(n) $}.
\STATE $n = 0$;
\STATE $\forall n_k \in \lbrace 1, \ldots, N_k \rbrace$,
\begin{eqnarray}
\nonumber
\hat{v}_{k,A_k^{(n_k)}}(0) & = & 0,\\
\nonumber
\hat{\pi}_{k,A_k^{(n_k)}}(0) & = & \frac{1}{N_k}.
\end{eqnarray}
\STATE $a_k(0) \sim \hat{\bs{\pi}}_k(0)$;
\FORALL {$n > 0$}
\STATE $\forall n_k \in \lbrace 1, \ldots, N_k \rbrace$, update $\hat{\bs{\pi}}_k(n)$.
\STATE
\begin{eqnarray}\nonumber
	a_k(n) =\left\lbrace
	\begin{array}{ll}
	a_k(n-1) & \text{ if } \tilde{v}_{k}(n-1) = 1  \\
	a_k(n) \sim \hat{\bs{\pi}}_k(n) & \text{ otherwise. } 
	\end{array} 						
	\right.
\end{eqnarray}
\ENDFOR \textbf{end}
\end{algorithmic}
\caption{Learning the SE of the Game $\GameSF$ [Player $k \in \mathcal{K}$]}\label{Algo1}
\end{algorithm}
\begin{proposition}\label{PropConvergenceAlgo1}\emph{
The behavioral rule  \eqref{EqActionAtIntervalN} with probability distributions $\bs{\pi}_k = \left( \pi_{k,A_k^{(1)}}, \ldots, \pi_{k,A_k^{(N_k)}} \right) \in \triangle\left( \mathcal{A}_k\right)$, with $k \in \mathcal{K}$,
converges to a SE of the game $\GameSF$ in finite time if for all $k \in \mathcal{K}$ and for all $n_k \in \lbrace 1, \ldots, N_k \rbrace$, it holds that,
\begin{equation}\label{EqNonZeroProb}
\pi_{k,A_k^{(n_k)}}(n) >  0,
\end{equation}
at each time interval $n \in \mathds{N}$, and assumptions $(i)$, $(ii)$ and $(iii)$ always hold.
}
\end{proposition}
The proof of Prop. \ref{PropConvergenceAlgo1} follows simply from the fact that \eqref{EqNonZeroProb} implies that every action profile will be played at least once with non-zero probability during a sufficiently large time interval. Then, since at least one SE exists, this action profile will be played at least once. Now, from \eqref{EqActionAtIntervalN}, it follows that once a SE is played, no player changes its current action. Thus, convergence is observed. 
 
\noindent
From the reasoning above, it can be concluded that any probability distribution $\hat{\bs{\pi}}_{k}(n)$ such that all actions have a non-zero probability of being played, for all $n$, can be chosen as the probability distribution of exploration. However, the choice of this probability distributions might impact the convergence time. Two particular ways for building the probability distribution $\hat{\bs{\pi}}_k(n)$ are proposed in  \cite{Ross-07}. In the first case, a uniform probability distribution during the whole learning process was used. That is, for all $k \in \mathcal{K}$ and for all $n_k \in \lbrace 1, \ldots, K \rbrace$,
\begin{equation}\label{EqUniformDistrib}
\hat{\pi}_{k,A_k^{(n_k)}}(n) = \frac{1}{N_k}.
\end{equation}
In the second case, at time interval $n$, higher probabilities are assigned to actions which have been played a smaller number of times during all time intervals between $0$ and $n-1$. Let $T_{k,A_k^{(n_k)}}(n) \in \mathds{N}$, with $k \in \mathcal{K}$ and $n_k \in \lbrace 1, \ldots, N_k\rbrace$, be the number of times that player $k$ has played action $A_k^{(n_k)}$ up to time interval $n$, i.e.,
\begin{equation}\label{EqT}
T_{k,A_k^{(n_k)}}(n) = \ds\sum_{s = 0}^{n-1}\ind{a_k(s) = A_k^{(n_k)}}.
\end{equation}
Then, the probability distribution to select the next action is the following:
\begin{equation}\label{EqRossDistribution}
    \hat{\pi}_{k,A_k^{(n_k)}}(n) = \frac{\frac{1}{T_{k,A_k^{(n_k)}}(n)}}{{\ds\sum_{m = 1}^{N_k} \frac{1}{T_{k,A_k^{(m)}}(n)}}},
\end{equation}
where $T_{k,A_k^{(n_k)}}(0) = \delta$, with $\delta > 0$.

\subsection{Clipping Strategies and SE}
The behavioral rule \eqref{EqActionAtIntervalN} converges to a SE in pure strategies in finite time. However, in real system scenarios, it is often observed that there might exists an action from a given player, which achieves satisfaction regardless of the actions adopted by all the other players. We refer to this kind of actions as \emph{clipping actions} \cite{Perlaza-Globecom-2010}.
\begin{definition}[Clipping Action]\label{DefClippingAction} \emph{In the game $\GameSF$, a player $k \in \mathcal{K}$ is said to have a clipping action $a_k\in \mathcal{A}_k$ if
\begin{equation}
\forall \bs{a}_{-k} \in \mathcal{A}_{-k}, \quad a_k \in f_{k}\left(\bs{a}_{-k}\right).
\end{equation}}
\end{definition}
As shown in the following proposition, the existence of clipping actions in the game $\GameSF$ might inhibit the convergence of the behavioral rule in \eqref{EqActionAtIntervalN}.
\begin{proposition} \label{PropNonConvergenceSESA} \emph{
Consider the game $\GameSF$  in satisfaction form.
Assume the existence of at least one clipping action and denote it by $a_{k}^* \in \mathcal{A}_k$ for player $k$, with $k \in \mathcal{K}$. Then, if there exists a player $j \in \mathcal{K}\setminus\lbrace k \rbrace$, for which $f_j\left(a_k^*, \bs{a}_{-\lbrace j,k\rbrace}\right) = \emptyset $, $\forall \bs{a}_{-\lbrace j,k\rbrace} \in \ds\prod_{i \in \mathcal{K}\setminus\lbrace {j,k} \rbrace} \mathcal{A}_i$. Then, the behavioral rule in \eqref{EqActionAtIntervalN} does not converge to a SE with strictly positive probability.}
\end{proposition}
The proof of Prop. \ref{PropNonConvergenceSESA} follows from the fact that at time $n >0$ before convergence, the probability that player $k$ plays the clipping action $a_k^*$ is strictly positive \eqref{EqNonZeroProb}. If player $k$ plays $a_k^*$, by definition, there exist a player $j \neq k$ which would never be satisfied. Then, the behavioral rule does not converge to any SE. Nonetheless, simple alternatives can be used to solve this problem. For instance, the behavioral rule in \eqref{EqActionAtIntervalN} can be modified such that a player changes its current action (using a given probability distribution over the actions) even if it is satisfied when it sees the other players not satisfied during a long period. Nonetheless, in this case, players would need to have more than 1-bit feedback in order to detect the non-satisfaction of the others. For instance, the feedback of the instantaneous value of the metric performance. This approach can be compared with the idea of epsilon experimentation discussed in \cite{Young-2004}.
\section{Applications}\label{SecApplications}

In this section, we apply the concept of SE to the case of a classical interference channel \cite{Carleial1978} with $2$ pairs of transmitter-receiver pairs sharing a common bandwidth. Here, the notion of SE are compared with the existing equilibrium notions such as GNE. At the same time, the performance of the behavioral rules presented in Sec. \ref{SecLearningSE} is evaluated in terms of convergence time to a satisfaction equilibrium.

\subsection{Satisfaction Equilibrium in the Interference Channel}

Consider a set $\mathcal{K} = \lbrace 1, 2 \rbrace$  of two transmitter-receiver pairs simultaneously operating over the same frequency band and thus, subject to mutual interference. Each transmitter communicates only with its corresponding receiver and any kind of message exchange aiming to achieve transmit cooperation is not considered. For all $(j,k) \in \mathcal{K}^2$,  denote by $g_{j,k}$ and $p_{k}^{(n_k)}$ the channel gain between transmitter $k$ and receiver $j$, and the $n_k$-th transmit power level of transmitter $k$, respectively.
We denote by $\mathcal{A}_k = \left\lbrace  p_k^{(1)}, \ldots, p_k^{(N_k)} \right\rbrace$, the set of all possible transmit power levels of player $k$. For all $k \in \mathcal{K}$, the minimum transmit power is $p_k^{(1)} = 0$ and the maximum transmit power is $p_k^{(N_k)} = p_{k,\max}$.
The QoS metric, denoted by $u_k: \mathcal{A}_1 \times \mathcal{A}_2 \rightarrow \mathds{R}_+$, of the transmitter-reciever pair $k$ is its (Shannon) transmission rate in bits per second (bps). Thus, for all $\left( p_k, p_{-k}\right) \in \mathcal{A}_{k} \times \mathcal{A}_{-k} $, we write that,
\begin{equation}
	u_{k}\left( p_k, p_{-k} \right)= \log_{2}\left(1 + \frac{p_{k} g_{k,k}}{\sigma^2_k +   p_{-k} g_{k,-k}} \right)\mbox{ [bps/Hz]}.
\end{equation}
 Here, $\sigma_k^2$ is the noise level at receiver $k$ and we denote the signal to noise ratio at the transmitter $k$ by $\mathrm{SNR}_k = \frac{p_{k,\max}}{\sigma_k^2}$.
The QoS requirement for player $k$ is to provide a transmission rate higher than $\Gamma_k$ bps. Thus, we model the satisfaction correspondence $f_k$, as follows,
\begin{equation}\label{EqFkExample}
f_k\left( \bs{p}_{-k} \right) =  \left\lbrace p_k \in \mathcal{A}_k: u_{k}\left( p_k, p_{-k} \right) \geqslant \Gamma_k \right\rbrace.
\end{equation}
We assume also that transmitters associate different effort measures with each of their power levels. The higher the transmit power, the higher the effort.

\noindent
This scenario is modeled by a game
in satisfaction form $\GameSF$ 
and a game in normal form with constrained action sets $\GameSEprime$, where, for all $k \in \mathcal{K}$, the cost or effort function $c_k$ is defined as follows
\begin{equation}
c_k \left( p_k \right) = \left\lbrace \begin{array}{lcl}
p_{k,\max} + \delta & \text{ if } &  p_k = p_{k}^{(1)}\\
p_k & \text{ if } & p_k \in \lbrace p_{k}^{(2)}, \ldots, p_{k}^{(N_k)} \rbrace,
\end{array}
\right.
\end{equation}
where $\delta > 0 $.
Note that the most costly action is not to transmit. This choice is made to force the radio devices to transmit any time it is possible.
The correspondence $f_k'$ is defined as follows:
\begin{equation}\label{EqFkprime}
f_k' \left( \bs{p}_{-k} \right) = \left\lbrace
\begin{array}{rcl}
  f_k \left( \bs{p}_{-k} \right) & \text{if} & f_k' \left( \bs{p}_{-k} \right) \neq \emptyset \\
   p_k^{(1)} & \text{otherwise} &.
  \end{array}
\right.
\end{equation}

\noindent
Here, we include the non-transmission action $p_k^{(1)} =0$ in order to avoid an empty set of actions for players $k$, when it does not exists an action able to achieve the required minimum rate.

\noindent
Note that if the following holds, $\forall k \in \mathcal{K},$ $\exists p_k \in \mathcal{A}_k$, such that
\begin{equation}
\log_2\left( 1 + \frac{p_{k,} g_{k,k}}{\sigma^2_k + p_{-k,\max} g_{k,-k}} \right)> \Gamma_k,
\end{equation} 
the existence of at least one SE in ensured from Theorem \ref{TheoExistence2}. Note that under such an assumption $f_k(\bs{a}_{-k})$ is never empty (condition ($ii$)) and by establishing the comparison \emph{bigger or equal than} ($\geqslant$) as the binary relation $\preceq$ in $\mathcal{A}$, the condition ($i$) and ($iii$) are always verified. 

\begin{figure}
\centering
\includegraphics[width=\linewidth]{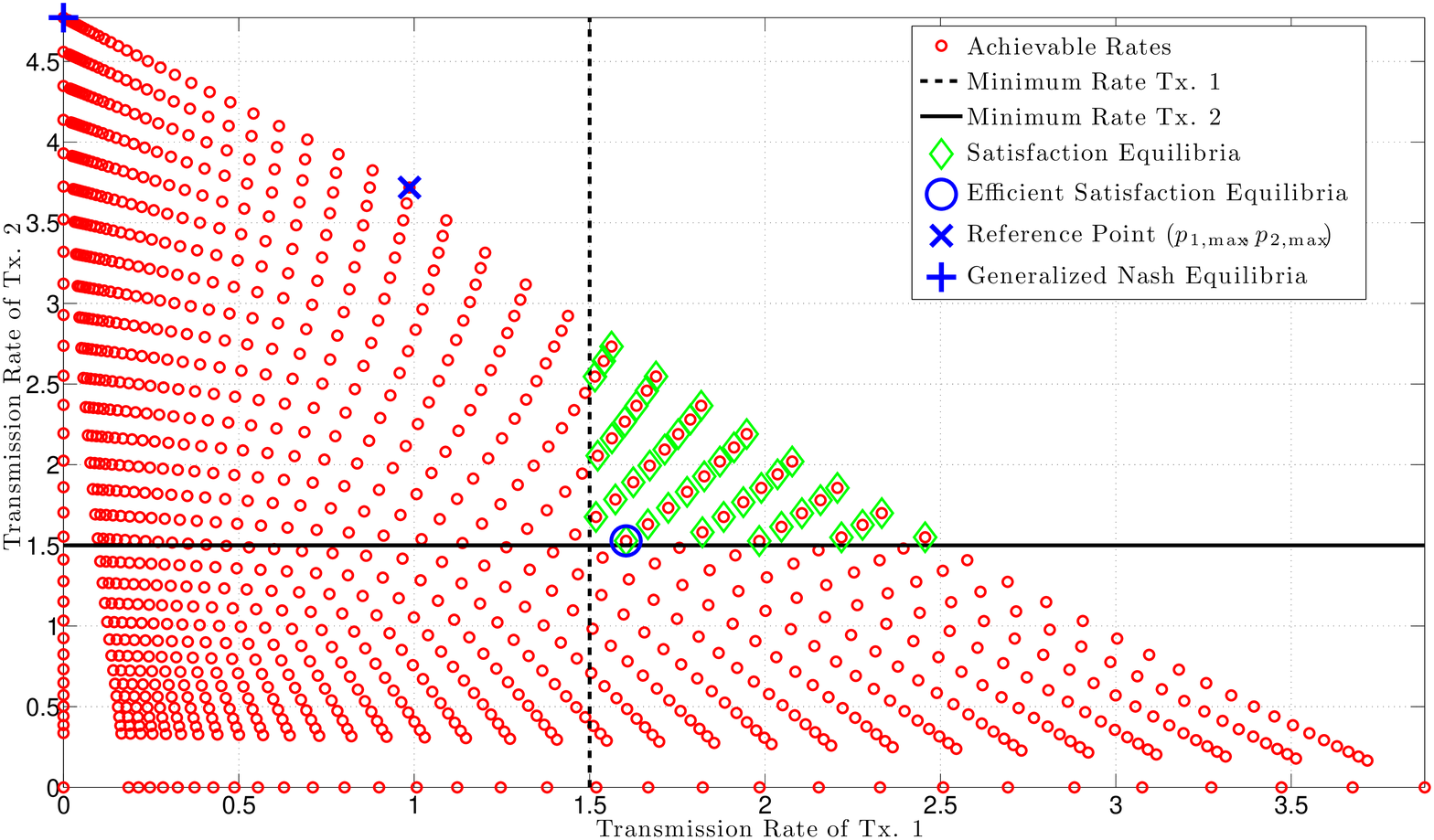}
\caption{Achievable (Shannon) transmission rates $\left( u_1(p_1,p_2), u_2(p_1,p_2) \right)$, for all $\left( p_1, p_2 \right) \in \mathcal{A}_1 \times \mathcal{A}_2$, with  $\mathrm{SNR} = \frac{p_{k,\max}}{\sigma^2_k} = 10$ dBs, $\left( \Gamma_1, \Gamma_2 \right) = \left(1.5 , 1.5 \right)$ bps and $N_1 = N_2 = 32$ power levels. }
\label{FigUtilities}
\end{figure}

\noindent
In Fig. \ref{FigUtilities}, we plot (in red circles) all the achievable (Shannon) transmission rates for both transmitters, i.e., the pairs $\left( u_1(p_1,p_2), u_2(p_1,p_2) \right)$, for all $\left( p_1, p_2 \right) \in \mathcal{A}_1 \times \mathcal{A}_2$ and a particular channel realization.  All the equilibria of the games 
$\gameSF$ and $\gameSEprime$  are plotted.
The game $\gameSEprime$ has a unique equilibrium which is the pair $\left( 0, p_{2,\max} \right)$ (see (Def. \ref{DefGNE}) and reference point $\left( p_{1,\max} , p_{2,\max} \right)$ in Fig. \ref{FigUtilities}).
The game $\gameSF$ has multiple equilibria (Def. \ref{DefSE}). In particular, note that 
with the game in normal form with constrained strategies, it is not possible to simultaneously satisfy the QoS of both transmitters. In this case, only transmitter $2$ can be satisfied.  On the contrary, in the game formulated in satisfaction form, all players are able to satisfy their QoS demands at the equilibrium of the game.  Importantly, the ESE satisfies the QoS condition for both transmitters with the lowest transmit power, while all the other SE requires a higher transmission power. In particular, note that (with this particular channel realization) the set of GNE and ESE appear to be unitary. However, as shown before, the existence and uniqueness of the ESE and GNE are conditioned in general. With this simple example, we have shown that by including the notion of performance maximization, i.e., the notion of GNE, leads to an unsatisfying game outcome, where only one player is satisfied, while the simultaneous satisfaction of both players is feasible.
\subsection{Clipping Actions in the Interference Channel}
Note that the game $\gameSF$ with the particular channel realization used in Fig. \ref{FigUtilities} possess at least one clipping action. For instance, when transmitter $2$ transmits at the maximum power $p_{2,\max}$, it is always satisfied even if player $1$ transmits at the maximum power (see the reference point $\left( p_{1,\max} , p_{2,\max} \right)$ in Fig. \ref{FigUtilities} ).  At the same time, if player $2$ transmits at the maximum power, player $1$ is unable to achieve satisfaction. Hence, if before observing convergence, transmitter $2$ uses, for instance, its maximum transmit power, then convergence to a SE is not observed neither in finite nor infinite time.

\begin{figure}
\centering
\includegraphics[width=\linewidth]{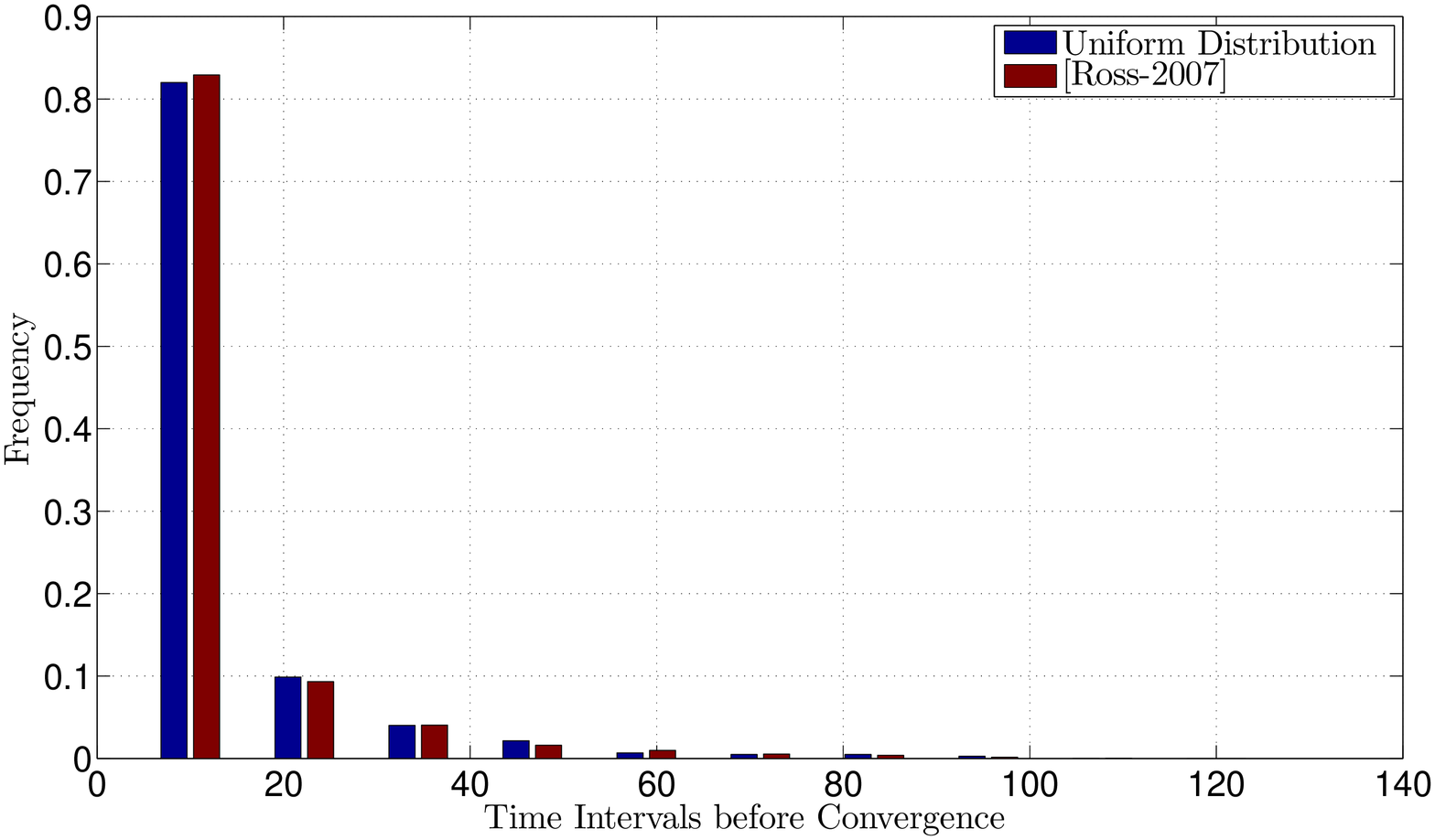}
\caption{Histogram of the convergence time to a SE in the game $\GameSF$ using the algorithm (Alg. \ref{Algo1}).  Here,  $\mathrm{SNR} = \frac{p_{k,\max}}{\sigma^2_k} = 10$ dBs, $\left( \Gamma_1, \Gamma_2 \right) = \left(1.5,1.5 \right)$ bps and $N_1 = N_2 = 32$ power levels.
}
\label{FigHistogramTimeConvergence}
\end{figure}
\begin{figure}
\centering
\includegraphics[width=\linewidth]{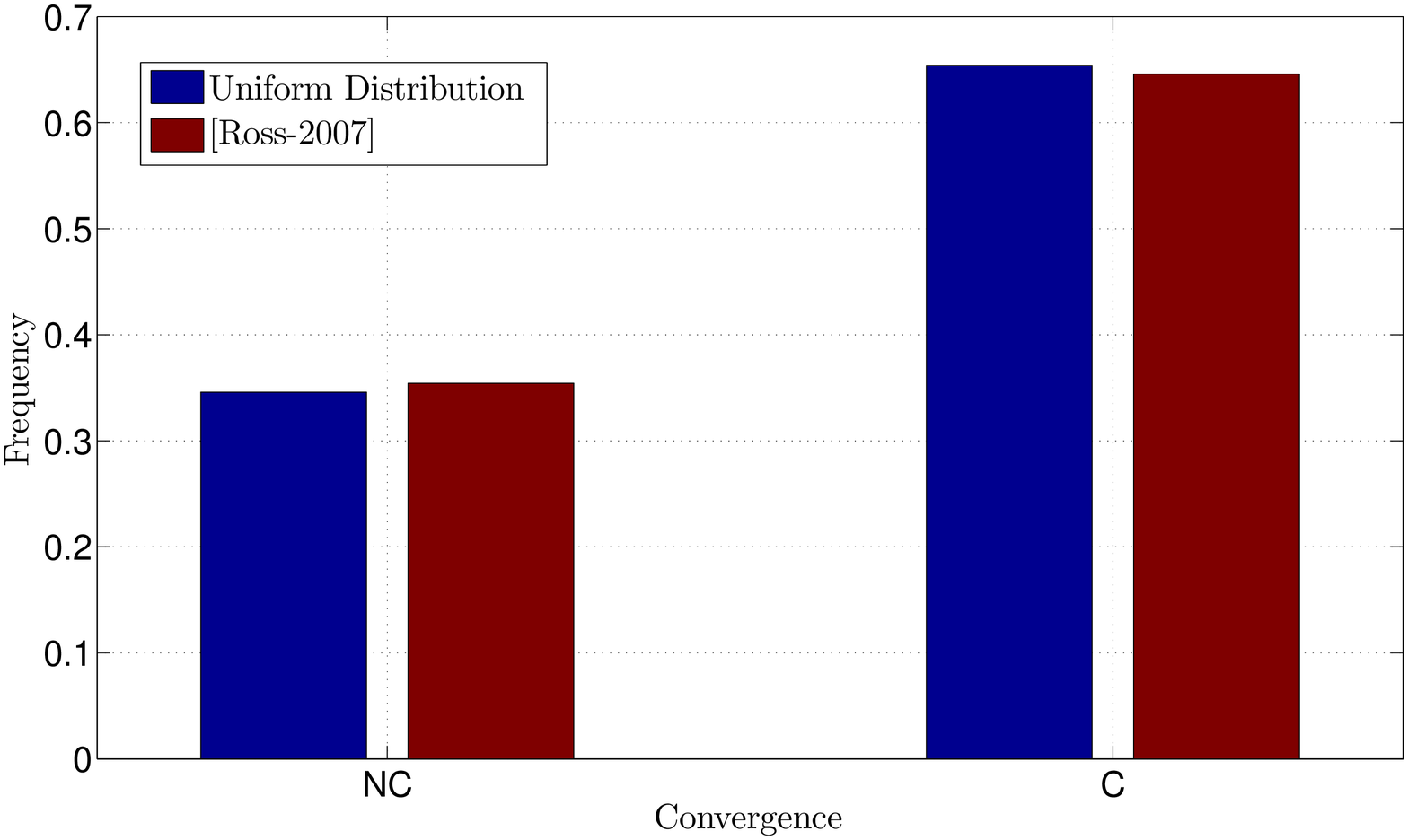}
\caption{Histogram of the event of convergence or non-convergence  of the learning algorithm (Alg. \ref{Algo1})  in the game $\GameSF$. Here,  $\mathrm{SNR} = \frac{p_{k,\max}}{\sigma^2_k} = 10$ dBs, $\left( \Gamma_1, \Gamma_2 \right) = \left(1.5,1.5 \right)$ bps and $N_1 = N_2 = 32$ power levels.}
\label{FigHistogramTimeNonConvergence}
\end{figure}

\noindent
In Fig. \ref{FigHistogramTimeNonConvergence}, we show a histogram of the convergence or not convergence of the algorithm.  Here, we say that the algorithm does not convergence if during $100$ consecutive time intervals, a given player does not change its current action, while the other still does (this implies that a clipping action might be being played).
At each trial of the algorithm, we use the same channel realization used in Fig. \ref{FigUtilities}.
Note that independently of the probability distributions $\hat{\bs{\pi}}_k(n)$ adopted by player $k$ to try new actions, the event of one player playing a clipping strategy is non-negligible ($0.3$). In the particular case of the interference channel as treated here, the corresponding game is free of clipping actions if the simultaneous transmission at maximum power allows satisfaction. However, in this case, the distinction between SE and NE looses its importance since both equilibrium concepts would be able to give a satisfactory solution to the QoS problem. This observation leaves open the way for further research on learning algorithms in the context of the SE in the presence of clipping actions.

\subsection{Convergence Time to the SE}
Now, our interest focuses on the average time for converging to one SE of the game $\gameSF$, when convergence is observed in the previous experiment.
The convergence time is measured as the number of action updates required to each transmitter before convergence.
\begin{figure}
\centering
\includegraphics[width=\linewidth]{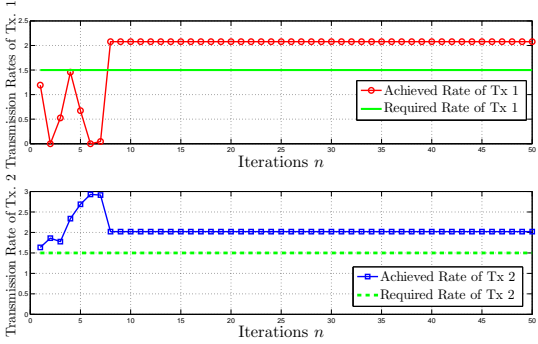}
\caption{Instantaneous achieved rates of transmitter $1$ (red) and $2$ (blue).Here,  $\mathrm{SNR} = \frac{p_{k,\max}}{\sigma^2_k} = 10$ dBs, $\left( \Gamma_1, \Gamma_2 \right) = \left(1.5,1.5 \right)$ bps and $N_1 = N_2 = 32$ power levels.}
\label{FigConvergence}
\end{figure}
In Fig. \ref{FigHistogramTimeConvergence}, we show a histogram of the convergence time when players try new actions with the probability distribution in \eqref{EqUniformDistrib} and \eqref{EqRossDistribution}. Note that in this particular scenario, using a probability different from the uniform distribution does not bring a significant improvement. Interestingly, the histogram shows that if convergence is observed, most of the time ($80\%$), satisfaction is achieved in less than $20$ time intervals (action updates).

\noindent
In Fig. \ref{FigConvergence}, we plot the achieved transmission rate of both links at each instant $n$ when the behavioral rule  \eqref{EqActionAtIntervalN} is used. Therein, it can be observed that even though a transmitter is satisfied, and thus does not change its transmission power level, its instantaneous transmission rate changes due to the action updates of the other transmitters. Once both transmitters are satisfied, then, none of them changes its transmit powers.

\section{Conclusions}\label{SecConclusions}
The game formulation in satisfaction form (SF) and the notion of satisfaction equilibrium (SE) introduced in this paper have been shown to be neatly adapted to model the problem of QoS provisioning in decentralized self-configuring networks. At the SE, all players are satisfied. On the contrary,
 when the QoS provisioning problem is modeled by games in classical normal form or normal form with constrained set of actions, equilibria where not all the players achieves satisfaction might be observed, even when there exist action profiles that allow the simultaneous satisfaction of all players.   The notion of SE has been formalized in the context of pure and mixed strategies and its existence and uniqueness has been studied. In particular, when no SE exists neither in pure nor mixed strategies, necessary and sufficient conditions for the existence of an epsilon-SE has been presented. However, not all games in SF possess an $\epsilon$-SE.  Finally,  a learning dynamics has been proposed to achieve SE.  In particular, we remark that it requires only $1$-bit feedback messages between the corresponding transmitter-receiver pairs. Nonetheless, the convergence remains still conditioned. This suggests that the design of algorithms such that at least one SE is learned in finite time and in a fully distributed fashion remains being an open problem.
\appendices
\section{ Potential Games with Constrained Set of Actions }\label{AppPGCSA}

In this appendix we present a generalization of a class of games known as exact potential games (PG)\cite{Monderer-Shapley-1996}.
We refer to this new class of games as constrained exact potential games. First, consider a game in normal form with constrained strategies and denote it by $\GameESFapp$. Let the set  $\mathcal{F}_k \subset \mathcal{A}$ be the graph of the correspondence $g_k$, hence,
\begin{equation}\label{EqGraphOffk}
\mathcal{F}_k = \left\lbrace  \left( a_k,\bs{a}_{-k} \right) \in \mathcal{A}: a_k \in g_k\left( \bs{a}_{-k} \right)\right\rbrace.
\end{equation}
The set $\mathcal{F}_k$ determines the action profiles which can be observed as outcomes of the game $\gameESFapp$, when only player $k$ is allowed to play given any action profile $\bs{a}_{-k}$ for which the set $g_k\left( \bs{a}_{-k}\right)$ is not empty. Following this reasoning, the set of  all possible outcomes of the game $\gameESFapp$ corresponds to the following set
\begin{equation}\label{EqIntersectionOfGraphs}
\mathcal{F} = \ds\bigcap_{j = 1}^{K} \mathcal{F}_j,
\end{equation}
which is the set of action profiles such that $\forall \bs{a} \in \mathcal{F}$, it holds that $\forall k\in\mathcal{K}$, $a_k \in g_k\left( \bs{a}_{-k} \right)$. Unilateral deviations of a set of players from any action profile $\bs{a} \in \mathcal{F}$ might lead to action profiles which do not belong to $\mathcal{F}$.
The following set
\begin{equation}\label{EqUnionOfGraphs}
\wideparen{\mathcal{F}} = \ds\bigcup_{j = 1}^{K} \mathcal{F}_j,
\end{equation}
contains all possible unilateral deviations one can observe from any action in the set $\mathcal{F}$. Using both sets $\mathcal{F}$ and $\wideparen{\mathcal{F}}$, we introduce the definition of exact constrained potential game.
\begin{definition}[Exact Constrained PG (ECPG)]\label{DefECPG} \emph{Any game in normal form with constrained set of actions $\GameESFapp$ is an exact constrained potential game (ECPG) if there exists a function $\phi: \wideparen{\mathcal{F}} \rightarrow \mathds{R}$ such that  for all $\bs{a} \in \wideparen{\mathcal{F}}$, it holds that, for all $k \in \mathcal{K}$ and for all $a_k' \in g_k\left( \bs{a}_{-k} \right)$,
\begin{equation}\nonumber
u_k(a_{k},\bs{a}_{-k}) - u_k(a'_{k},\bs{a}_{-k}) = \phi(a_{k},\bs{a}_{-k}) - \phi(a'_{k},\bs{a}_{-k}).
\end{equation}
}
\end{definition}
Before we continue, we clearly state that not all the the properties of potential games \cite{Monderer-Shapley-1996} hold for the constrained potential games. For instance, not all exact constrained PG have an equilibrium. In the following, we introduce two results regarding the existence of an equilibrium in pure strategies in ECPG.
\begin{theorem}[Existence of an equilibrium in ECPG]\label{TheoExistencePureStrategiesECPG}\emph{
The finite exact constrained potential game $\GameESFapp$, with potential function $\phi: \wideparen{\mathcal{F}} \rightarrow \mathds{R}_+$,
has at least one equilibrium in pure strategies, if the sets  $\mathcal{F}_1, \ldots, \mathcal{F}_K$, are non-empty and identical.
}
\end{theorem}
\begin{IEEEproof}\emph{
By assumption, the set $\mathcal{F}$ \eqref{EqIntersectionOfGraphs} is non-empty. Thus, there exists at least one feasible outcome $\bs{a}^* \in \mathcal{F}$ for the game $\mathcal{G}$. Now, for all $k \in \mathcal{K}$, any unilateral deviation of player $k$ from an action profile $\bs{a}^*$ leads to an action profile of the form $\left( a_k, \bs{a}_{-k}^*\right) \in \wideparen{\mathcal{F}}$. Similarly, by assumption, both sets $\mathcal{F}$ and $\wideparen{\mathcal{F}}$ are identical, thus, any unilateral deviation from a feasible action profile is also a feasible action profile.
Now,  without any loss of generality, let the elements of the sets $\mathcal{F} = \left\lbrace \tilde{A}^{(1)}, \ldots, \tilde{A}^{(N)}\right\rbrace$ be indexed following any particular order such that the following holds,
\begin{equation}
\phi\left( \tilde{A}^{(1)} \right) \leqslant \phi\left( \tilde{A}^{(2)} \right) \leqslant  \ldots \leqslant  \phi\left(\tilde{A}^{(N)} \right),
\end{equation}
 with $N = \left| \mathcal{F} \right|$.
Thus, from  Def. \ref{DefGNE}, it holds that $\tilde{A}^{(N)}$ is an equilibrium of the game $\mathcal{G}$ , which completes the proof.
}
\end{IEEEproof}%

\section{Multiplicity of the ESE}\label{AppMultiplicityESE}

In this appendix,  we study the multiplicity of the equilibrium of the potential game with constrained strategies $\GameESF$. For doing so, we take advantage of the fact that it is a potential game with constrained action sets and thus, we analyze the auxiliary game $\GameESFpotential$. Note that this choice does not imply any lost of generality, since the set of GNE of both games are identical.

\noindent
In the following, we use some tools from graph theory to determine the number of ESE which a given game in efficient-satisfaction form $\GameESF$ can  possess.
We start by indexing the elements of the action set $\mathcal{A}$ in any given order using the index $n \in \mathcal{I} = \left\lbrace 1, \ldots, \left| \mathcal{A} \right| \right \rbrace$. Denote by $\bs{a}^{(n)} = \left(a^{(n)}_1, \ldots, a^{(n)}_K\right)$ the $n$-th element of the action set $\mathcal{A}$.
Consider that each action profile $\bs{a}^{(n)}$ is associated with a vertex $x_n$ in a given directed graph $G$.  There exists an arc from vertex $x_n$ to another vertex $x_m$, if the action profile represented by the latter $\bs{a}^{(m)}$ can be obtained from the former $\bs{a}^{(n)}$ by changing the action of only one player and lower potential (sum of efforts) is obtained. For instance, if the unique deviator is player $k$, then,  $a^{(m)}_k \in f_k\left( \bs{a}_{-k}^{(n)}\right)$  and $\phi\left( \bs{a}^{(n)} \right) > \phi\left( \bs{a}^{(m)} \right)$.
More precisely, the graph $G$ can be defined by the pair $G = \left( \mathcal{X}, \bs{B}\right)$, where the set $\mathcal{X} = \left\lbrace x_1, \ldots, x_{\left|\mathcal{A}\right|}\right\rbrace$ (nodes) contains the nodes representing the action profiles in the set $\mathcal{A}$ and $\bs{B}$ (edges) is
a non-symmetric matrix with dimensions $\left|\mathcal{A}\right| \times \left|\mathcal{A}\right|$ and entries defined as follows $\forall (n,m) \in \mathcal{I}^2$ and  $n\neq m$,

\begin{equation}\label{EqMatrixAHat}
    b_{n,m} = \left\lbrace \begin{array}{ccl} 1 & \text{ if }        & (i) \; \exists! \, k \in \mathcal{K}: \, a_k^{(n)} \neq a_k^{(m)}, \\
    & & (ii)a_k^{(m)} \in f_k\left( a_{-k}^{(n)} \right)\\
    & & (iii) \; \phi\left(\bs{a}^{(m)}\right) < \phi\left(\bs{a}^{(n)}\right)\\
                                                                                                    0 & \text{ otherwise }, &
                                                                     \end{array} \right.
\end{equation}
and $b_{i,i} = 0$ for all $i \in \mathcal{I}$.

\noindent
A realistic assumption is to consider that for any pair of action profiles $\bs{a}^{(n)}$ and $\bs{a}^{(m)}$ which are adjacent, we have that $\phi(\bs{a}^{(n)}) \neq  \phi(\bs{a}^{(m)})$. This is because players assign different effort values to their actions.
From the definition of the matrix $\bs{B}$, we have that a necessary and sufficient condition for a vertex $x_n$ to represent an ESE action profile is to have a null out-degree in the oriented graph $G$, i.e., there are no outgoing edges from the node $x_n$ (sink vertex).
Finally, one can conclude that determining the set of ESE in the game $\GameESF$ boils down to identifying all the sink vertices in the oriented graph $G$.
That is, exploiting the fact that, if the $n$-th row of the matrix $\bs{B}$ is a null-vector, then the action $\bs{a}^{(n)}$ is an ESE of the game $\tilde{\mathcal{G}}$.
Interestingly, a particular case arises when the resulting graph is an edgeless graph, i.e., the corresponding matrix $\bs{B}$ is a null matrix. In this case, the set of SE would be identical to the set of ESE, which implies that the idea of associating an effort to each action is not enough to select an ESE among the set of SE.
In any case, determining the exact set of SE would require the analysis of the matrix $\bs{B}$, which might be highly  demanding and
requires complete information.
\bibliographystyle{IEEEtran}
\bibliography{GT}

\begin{thebibliography}{10}
\providecommand{\url}[1]{#1}
\csname url@samestyle\endcsname
\providecommand{\newblock}{\relax}
\providecommand{\bibinfo}[2]{#2}
\providecommand{\BIBentrySTDinterwordspacing}{\spaceskip=0pt\relax}
\providecommand{\BIBentryALTinterwordstretchfactor}{4}
\providecommand{\BIBentryALTinterwordspacing}{\spaceskip=\fontdimen2\font plus
\BIBentryALTinterwordstretchfactor\fontdimen3\font minus
  \fontdimen4\font\relax}
\providecommand{\BIBforeignlanguage}[2]{{%
\expandafter\ifx\csname l@#1\endcsname\relax
\typeout{** WARNING: IEEEtran.bst: No hyphenation pattern has been}%
\typeout{** loaded for the language `#1'. Using the pattern for}%
\typeout{** the default language instead.}%
\else
\language=\csname l@#1\endcsname
\fi
#2}}
\providecommand{\BIBdecl}{\relax}
\BIBdecl

\bibitem{Mackenzie-01}
A.~MacKenzie and S.~Wicker, ``Game theory and the design of self-configuring,
  adaptive wireless networks,'' \emph{IEEE Communications Magazine}, vol.~39,
  no.~11, pp. 126--131, Nov 2001.

\bibitem{Lasaulce-Tembine-Book-2011}
S.~Lasaulce and H.~Tembine, \emph{Game Theory and Learning in Wireless
  Networks: Fundamentals and Applications}.\hskip 1em plus 0.5em minus
  0.4em\relax Elsevier Academic Press, Oct. 2011.

\bibitem{Scutari-Palomar-2009}
G.~Scutari, D.~P. Palomar, J.-S. Pang, and F.~Facchinei, ``Flexible design of
  cognitive radio wireless systems: From game theory to variational inequality
  theory,'' \emph{IEEE Signal Processing Mag.,}, vol.~26, no.~5, pp. 107--123,
  Sept. 2009.

\bibitem{Scutari-Palomar-2010}
G.~Scutari, D.~P. Palomar, F.~Facchinei, and J.-S. Pang, ``Convex optimization,
  game theory, and variational inequality theory in multiuser communication
  systems,'' \emph{IEEE Signal Processing Magazine}, vol.~27, no.~3, pp.
  35--49, May 2010.

\bibitem{Facchinei-2007b}
F.~Facchinei, A.~Fischer, and V.~Piccialli, ``On generalized {N}ash games and
  variational inequalities,'' \emph{Oper. Res. Lett.}, vol.~35, no.~2, pp.
  159--164, 2007.

\bibitem{Facchinei-Pang-2009}
F.~Facchinei and J.-S. Pang, ``Nash equilibria: The variational approach,''
  \emph{Convex Optimization in Signal Processing and Communications, D. P.
  Palomar and Y. C. Eldar, Eds.}, vol. ch. 12, pp. 443--493, 2009.

\bibitem{Debreu-1952}
G.~Debreu, ``A social equilibrium existence theorem,'' \emph{Proceedings of the
  National Academy of Sciences of the United States of America}, vol.~38,
  no.~10, pp. 886--893, October 1952.

\bibitem{Rosen-65}
J.~B. Rosen, ``Existence and uniqueness of equilibrium points for concave
  n-person games,'' \emph{Econometrica}, vol.~33, no.~3, pp. 520--534, 1965.

\bibitem{Nash-1950}
J.~F. Nash, ``Equilibrium points in n-person games,'' \emph{Proceedings of the
  National Academy of Sciences of the United States of America}, vol.~36,
  no.~1, pp. 48--49, 1950.

\bibitem{PangScutari-IT-08}
J.-S. Pang, G.~Scutari, F.~Facchinei, and C.~Wang, ``Distributed power
  allocation with rate constraints in {G}aussian parallel interference
  channels,'' \emph{IEEE Trans. on Info. Theory}, vol.~54, no.~8, pp.
  3471--3489, Aug. 2008.

\bibitem{Altman-Altman-2003}
E.~Altman and Z.~Altman, ``S-modular games and power control in wireless
  networks,'' \emph{IEEE Transactions on Automatic Control}, vol.~48, no.~5,
  pp. 839--842, May 2003.

\bibitem{Fudenberg-Tirole-1991}
D.~Fudenberg and J.~Tirole, ``Game theory,'' \emph{MIT Press}, 1991.

\bibitem{Yin-Mehta-2011}
H.~Yin, U.~V. Shanbhag, and P.~G. Mehta, ``Nash equilibrium problems with
  scaled congestion costs and shared constraints,'' \emph{IEEE Transactions on
  Automatic Control}, vol.~56, no.~7, pp. 1702--1708, Jul. 2011.

\bibitem{Iiduka-2010}
H.~Iiduka, ``Fixed point optimization algorithm and its application to power
  control in cdma data networks,'' \emph{Mathematical Programming}, Oct. 2010.

\bibitem{Iiduka-2011}
------, ``Decentralized algorithm for centralized variational inequalities in
  network resource allocation,'' \emph{Journal of Optimization Theory and
  Applications}, pp. 1--16, Jun. 2011.

\bibitem{Yamada-2001}
I.~Yamada, ``The hybrid steepest descent for the variational inequality
  problems over the intersection of fixed point sets of nonexpansive
  mappings,'' \emph{Inherently Parallel Algorithms in Feasibility and
  Optimization and Their Applications. D. Butnariu, Y. Censor, and S. Reich,
  editors.}, no.~10, pp. 473--504, 2001.

\bibitem{LuX-2009}
X.~Lu, H.~K. Xu, and X.~Yin, ``Hybrid methods for a class of monotone
  variational inequalities,'' \emph{Nonlinear Analysis}, vol.~71, pp.
  1032--1041, 2009.

\bibitem{MarinoG-2011}
G.~Marino and H.~K. Xu, ``Explicit hierarchical fixed point approach to
  varitional inequalities,'' \emph{Journal of Optimization Theory and
  Applications}, vol. 149, pp. 61--78, 2011.

\bibitem{Rose-Perlaza-CommMag2011}
L.~Rose, S.~M. Perlaza, S.~Lasaulce, and M.~Debbah, ``Learning equilibria with
  partial information in wireless networks,'' \emph{IEEE Communication
  Magazine. Special Issue on Game Theory for Wireless Communications}, Aug.
  2011.

\bibitem{Ross-06}
S.~Ross and B.~Chaib-draa, ``Satisfaction equilibrium : Achieving cooperation
  in incomplete information games,'' in \emph{the 19th Canadian Conf. on
  Artificial Intelligence}, 2006.

\bibitem{Ross-07}
------, ``Learning to play a satisfaction equilibrium.'' in \emph{Workshop on
  Evolutionary Models of Collaboration}, 2007.

\bibitem{Perlaza-Globecom-2010}
S.~M. Perlaza, H.~Tembine, S.~Lasaulce, and M.~Debbah, ``Satisfaction
  equilibrium: A general framework for {QoS} provisioning in self-configuring
  networks,'' in \emph{the IEEE Global Communications Conference (GLOBECOM)},
  Miami, USA, Dec. 2010.

\bibitem{Altman-Solan-2009}
E.~Altman and E.~Solan, ``Constrained games: The impact of the attitude to
  adversaryÕs constraints,'' \emph{IEEE Transactions on Automatic Control},
  vol.~54, no.~10, pp. 2435--2440, Oct. 2009.

\bibitem{Borel-1921}
\'{E}mile Borel, ``La th\'{e}orie du jeu et les \'{e}quations \`{a} noyau
  sym\'{e}trique,'' \emph{Comptes Rendus de l'Acad\'{e}mie des Sciences}, vol.
  173, pp. 1304--1308, Sept. 1921.

\bibitem{Knaster-Tarski-1928}
B.~Knaster and A.~Tarski, ``Un th\'{e}or\`{e}me sur les fonctions
  d{'}ensembles,'' \emph{Ann. Soc. Polon. Math.}, vol.~6, pp. 133--134, 1928.

\bibitem{Border-1985}
K.~C. Border, \emph{Fixed Point Theorems with Applications to Economics and
  Game Theory}.\hskip 1em plus 0.5em minus 0.4em\relax New York, NY, USA:
  Cambridge University Press, 1985.

\bibitem{Fishburn-1970}
P.~Fishburn, \emph{Utility Theory for Decision Making}.\hskip 1em plus 0.5em
  minus 0.4em\relax John Wiley and Sons, 1970.

\bibitem{Morgestein-Neumann-44}
J.~V. Neumann and O.~Morgenstern, ``Theory of games and economic behavior,''
  \emph{Princeton University Press}, 1944.

\bibitem{Osborne-Rubinstein-Book}
M.~J. Osborne and A.~Rubinstein, \emph{A Course in Game Theory}.\hskip 1em plus
  0.5em minus 0.4em\relax The MIT Press, Jul. 1994.

\bibitem{TiroleBook1991}
D.~Fudenberg and J.~Tirole, \emph{Game Theory}.\hskip 1em plus 0.5em minus
  0.4em\relax {MIT Press}, October 1991.

\bibitem{Young-2004}
H.~P. Young, ``Strategic learning and its limits ({A}rne {R}yde memorial
  lectures sereis),'' \emph{Oxford University Press, USA}, 2004.

\bibitem{Carleial1978}
A.~Carleial, ``Interference channels,'' \emph{Information Theory, IEEE
  Transactions on}, vol.~24, no.~1, pp. 60--70, 1978.

\bibitem{Monderer-Shapley-1996}
D.~Monderer and L.~S. Shapley, ``Potential games,'' \emph{Games and Economic
  Behavior}, vol.~14, pp. 124--143, 1996.

\end{thebibliography}
\end{document}